%% file: main.tex
\documentclass[11pt]{article}
\input{files/macros}

 \usepackage{setspace}
\title{Near-Optimal Two-Pass Streaming Algorithm \\ for Sampling Random Walks over Directed Graphs}
\author{
Lijie Chen\thanks{MIT. \texttt{lijieche@mit.edu}}
\and
Gillat Kol\thanks{Princeton University. \texttt{gillat.kol@gmail.com}}
\and
Dmitry Paramonov\thanks{Princeton University. \texttt{dp20@princeton.edu}}
\and
Raghuvansh R. Saxena\thanks{Princeton University. \texttt{rrsaxena@princeton.edu}}
\and
Zhao Song\thanks{Institute for Advanced Study. \texttt{zhaos@ias.edu}}
\and
Huacheng Yu\thanks{Princeton University. \texttt{yuhch123@gmail.com}}
}

\begin{document}

 

 \begin{titlepage}
\maketitle

\input{files/abstract}
 \thispagestyle{empty}
 \end{titlepage}




\input{files/intro.tex}   
\input{files/tech.tex}    

\input{files/prelim.tex}  
\input{files/algo.tex}    
\input{files/lowb.tex}    

\section*{Acknowledgments} Lijie Chen is supported by an IBM Fellowship. Zhao Song is supported in part by Schmidt Foundation, Simons Foundation, NSF, DARPA/SRC, Google and Amazon AWS. We would like to thank Rajesh Jayaram for discussions on $\ell_1$ heavy hitters.

\bibliographystyle{alpha}
\bibliography{ref}

\appendix
\section*{Appendix}
\input{files/info.tex}    
\input{files/app-proofs.tex} 
\input{files/app-starting-vertex-oblivious-lowb} 
\end{document}

%% file: files/macros.tex
\usepackage{amsmath}
\usepackage{amsthm}
\usepackage{amssymb}
\usepackage{algorithm}
\usepackage{color}
\usepackage{pifont}
\usepackage{graphicx}
\usepackage{url}
\usepackage{algpseudocode}
\usepackage{bbm}
\usepackage{comment}
\usepackage{dsfont}
\usepackage{framed}
\usepackage{bm}
\usepackage{subcaption}

\usepackage{expl3}

\usepackage{tikz}

\graphicspath{{./figs/}{../figs/}}
\usepackage{mathtools}

 \usepackage[margin=1in]{geometry}

\usepackage{xspace}
\usepackage{hyperref}  
\hypersetup{colorlinks=true}
\hypersetup{linkcolor=[rgb]{.7,0,0}}
\hypersetup{citecolor=[rgb]{0,.7,0}}
\hypersetup{urlcolor=[rgb]{.7,0,.7}}

\usetikzlibrary{arrows,calc}

\newenvironment{proofof}[1]{\begin{proof}[{Proof of #1}]}{\end{proof}}

\allowdisplaybreaks 

\newcommand{\dout}{d_{\sf out}}
\newcommand{\din}{d_{\sf in}}
\newcommand{\Nout}{N_{\sf out}}
\newcommand{\Nin}{N_{\sf in}}
\newcommand{\sfX}{\mathsf{X}}
\newcommand{\sfY}{\mathsf{Y}}

\newcommand{\INDEX}{\mathsf{INDEX}}

\newcommand{\wt}{\widetilde}

\newcommand{\eps}{\varepsilon}
\newcommand{\N}{\mathbb{N}}
\newcommand{\R}{\mathbb{R}}

\renewcommand{\tilde}{\wt}

\renewcommand{\bar}{\overline}

\newcommand{\poly}{\mathrm{poly}}

\newcommand{\dist}{\mathrm{dist}}

\newcommand{\ie}{\textit{i}.\textit{e}.\@\xspace}

\DeclareMathOperator{\supp}{supp}
\DeclareMathOperator{\nnz}{nnz}

\newcommand{\calD}{\mathcal{D}}

\newcommand{\calX}{\mathcal{X}}
\newcommand{\calY}{\mathcal{Y}}
\newcommand{\calZ}{\mathcal{Z}}

\newcommand{\RW}{\mathsf{RW}}
\newcommand{\visit}{\textsf{visit}}

\newcommand{\Vheavy}{V_{\sf heavy}}
\newcommand{\Vlight}{V_{\sf light}}

\newcommand{\WT}{\widetilde}

\newcommand{\WTVheavy}{\WT{V}_{\sf heavy}}
\newcommand{\WTVlight}{\WT{V}_{\sf light}}

\newcommand{\calS}{\mathcal{S}}
\newcommand{\calH}{\mathcal{H}}

\newcommand{\algo}{\mathbb{A}}

\newcommand{\Lsave}{L^{\sf save}}

\newcommand{\prealgo}{\mathbb{P}}
\newcommand{\sampalgo}{\mathbb{S}}

\newcommand{\prealgosub}{\mathbb{P}_{\sf one\text{-}pass}}
\newcommand{\sampalgosub}{\mathbb{S}_{\sf one\text{-}pass}}

\newcommand{\prealgomain}{\mathbb{P}_{\sf two\text{-}pass}}
\newcommand{\sampalgomain}{\mathbb{S}_{\sf two\text{-}pass}}

\newcommand{\algotwopass}{\algo_{\sf two\text{-}pass}}
\newcommand{\algosub}{\algo_{\sf one\text{-}pass}}

\newcommand{\ustart}{u_{\sf start}}

\newcommand{\Vfull}{V_{\sf full}}
\newcommand{\Vsamp}{V_{\sf samp}}

\newcommand{\vecLsave}{\vec{L}^{\sf save}}
\newcommand{\bvecLsave}{\vec{\mathsf{L}}^{\sf save}}

\newcommand{\calE}{\mathcal{E}}
\newcommand{\bWTVheavy}{\WT{\mathsf{V}}_{\sf heavy}}
\newcommand{\bWTVlight}{\WT{\mathsf{V}}_{\sf light}}

\newcommand{\eg}{\textit{e}.\textit{g}.\@\xspace}

\usepackage[breakable]{tcolorbox}
\newtcolorbox{construction}[2][]
{
	breakable,
	colframe = gray!50,
	colback  = gray!10,
	coltitle = gray!10!black,
	before skip = 10pt,
	after skip = 10pt,
	title    = \textbf{#2},
	#1,
}
\newtcolorbox{graphview}[2][]
{
	breakable,
	colframe = black!30,
	colback  = black!0,
	coltitle = gray!10!black,
	before skip = 10pt,
	after skip = 10pt,
	title    = \textbf{#2},
	#1,
}

\DeclarePairedDelimiter{\card}{\lvert}{\rvert}

\DeclarePairedDelimiter{\paren}{\lparen}{\rparen}
\DeclarePairedDelimiter{\set}{\lbrace}{\rbrace}
\DeclarePairedDelimiter{\sq}{[}{]}

\DeclarePairedDelimiter{\norm}{\lVert}{\rVert}

\DeclarePairedDelimiter{\tvdbasic}{\lVert}{\rVert}
\makeatletter
\newcommand{\@tvdstar}[2]{\tvdbasic*{#1 - #2}_{\mathrm{TV}}}
\newcommand{\@tvdnostar}[3][]{\tvdbasic[#1]{#2 - #3}_{\mathrm{TV}}}
\newcommand{\tvd}{\@ifstar\@tvdstar\@tvdnostar}
\makeatother

\newenvironment{reminder}[1]{\smallskip
	\noindent {\bf Reminder of #1.  }\em}{\smallskip}

\usepackage{aliascnt} 

\newtheorem{theorem}{Theorem}[section]
 \newtheorem*{theorem*}{Theorem}

 \newaliascnt{definition}{theorem}
 \newtheorem{definition}[definition]{Definition}
 \aliascntresetthe{definition}
 
 \newtheorem*{definition*}{Definition}

 \newaliascnt{lemma}{theorem}
 \newtheorem{lemma}[lemma]{Lemma}
 \aliascntresetthe{lemma}
 
 \newtheorem*{lemma*}{Lemma}

 \newaliascnt{claim}{theorem}
 \newtheorem{claim}[claim]{Claim}
 \aliascntresetthe{claim}
 
 \newtheorem*{claim*}{Claim}

 \newaliascnt{fact}{theorem}
 \newtheorem{fact}[fact]{Fact}
 \aliascntresetthe{fact}
 
 \newtheorem*{fact*}{Fact}

 \newaliascnt{observation}{theorem}
 \newtheorem{observation}[observation]{Observation}
 \aliascntresetthe{observation}
 
 \newtheorem*{observation*}{Observation}

 \newaliascnt{conjecture}{theorem}
 
 \aliascntresetthe{conjecture}
 
 \newtheorem*{conjecture*}{Conjecture}

 \newaliascnt{corollary}{theorem}
 \newtheorem{corollary}[corollary]{Corollary}
 \aliascntresetthe{corollary}
 
 \newtheorem*{corollary*}{Corollary}

 \newaliascnt{remark}{theorem}
 \newtheorem{remark}[remark]{Remark}
 \aliascntresetthe{remark}
 
 \newtheorem*{remark*}{Remark}

 \newaliascnt{proposition}{theorem}
 
 \aliascntresetthe{proposition}
 
 \newtheorem*{proposition*}{Proposition}

\makeatletter
\patchcmd{\ALG@step}{\addtocounter{ALG@line}{1}}{\refstepcounter{ALG@line}}{}{}
\newcommand{\ALG@lineautorefname}{Line}
\makeatother

\def\dist{\mathcal{D}}

\def\E{\mathop{{}\mathbb{E}}}
\def\H{\mathbb{H}}

\def\h{\mathsf{h}}
\def\I{\mathbb{I}}
\def\rv{\mathsf}

\def\B{\mathcal{B}}
\def\A{\mathcal{A}}

\def\e{\mathrm{e}}

\allowdisplaybreaks

%% file: files/abstract.tex
\begin{abstract}
For a directed graph $G$ with $n$ vertices and a start vertex $\ustart$, we wish to (approximately) sample an $L$-step random walk over $G$ starting from $\ustart$ with minimum space using an algorithm that only makes few passes over the edges of the graph. This problem found many applications, for instance, in approximating the PageRank of a webpage. If only a single pass is allowed, the space complexity of this problem was shown to be $\tilde{\Theta}(n \cdot L)$. Prior to our work, a better space complexity was only known with $\tilde{O}(\sqrt{L})$ passes. 
 
We settle the space complexity of this random walk simulation problem for {\em two-pass} streaming algorithms, showing that it is $\tilde{\Theta}(n \cdot \sqrt{L})$, by giving almost matching upper and lower bounds. Our lower bound argument extends to every constant number of passes $p$, and shows that any $p$-pass algorithm for this problem uses $\tilde{\Omega}(n \cdot L^{1/p})$ space. 
In addition, we show a similar $\tilde{\Theta}(n \cdot \sqrt{L})$ bound on the space complexity of any algorithm (with {\em any} number of passes) for the related problem of sampling an $L$-step random walk from {\em every} vertex in the graph.
\end{abstract}

%% file: files/intro.tex
\section{Introduction}

\subsection{Background and Motivation}

{\em \bf Graph streaming algorithms.} {\em Graph streaming algorithms} have been the focus of extensive study over the last two decades, mainly due to the important practical motivation in analyzing potentially {\em huge structured data} representing the relationships between a set of entities (\eg, the link graph between webpages and the friendship graph in a social network).  In the graph streaming setting, an algorithm gets access to a sequence of graph edges given in an arbitrary order and it can read them one-by-one in the order in which they appear in the sequence. The goal here is to design algorithms solving important graph problems that only make one or few passes through the edge sequence, while using as little memory as possible.

Much of the streaming literature was devoted to the study of {\em one-pass} algorithms and an~$\Omega(n^2)$ space lower bound for such algorithms was shown for many fundamental graph problems. A partial list includes: maximum matching and minimum vertex cover \cite{fkms+04,gkk12}, $s$-$t$ reachability and topological sorting \cite{cgmv20,fkms+04,hrr98},  shortest path and diameter \cite{fkms+04,fkms+09},  maximum and (global or $s$-$t$) minimum cut \cite{z11}, maximal independent set \cite{ack19_soda,cdk19}, and dominating set \cite{akl16,er14}.

Recently, the {\em multi-pass} streaming setting received quite a bit of attention. For some graph problems, allowing a few passes instead of a single pass can reduce the memory consumption of a streaming algorithm dramatically. In fact, even a single additional pass over the input can already greatly enhance the capability of the algorithms. For instance, minimum cut and $s$-$t$ minimum cut in undirected graphs can be solved in two passes with only $\tilde{O}(n)$ and $O(n^{5/3})$ space, respectively \cite{rsw18} (as mentioned above, any one-pass algorithm for these problems must use $\Omega(n^2)$ space). Additional multi-pass algorithms include an $O(1)$-pass algorithm for approximate matching \cite{gkms19,gkk12,k13,m05}, an $O(\log \log n)$-pass algorithm for maximal independent set \cite{ack19_soda,cdk19,ggkm+18}, and $O(\log n)$-pass algorithms for approximate dominating set \cite{akl16,cw16,himv16} and weighted minimum cut \cite{mn20}. 

{\em \bf Simulating random walks on graphs.} Simulating random walks on graphs is a well-studied algorithmic problem with may applications in different areas of computer science, such as connectivity testing \cite{r08}, clustering \cite{acl07,ap09,cop03,st13}, sampling \cite{jvv86}, generating random spanning tree \cite{s18}, and approximate counting \cite{js89}. Since most applications of random-walk simulation are concerned with huge networks that come from practice, it is of practical interest to design low-space graph streaming algorithms with few passes for this problem.

In an influential paper by Das Sarma, Gollapudi and Panigrahy~\cite{SarmaGP11}, an $\tilde{O}(\sqrt{L})$-pass and $\tilde{O}(n)$ space algorithm for simulating $L$-step random walks on directed graphs was established. (Streaming algorithms with almost linear space complexity, like this one, are often referred to as {\em semi-streaming} algorithms). Using this algorithm together with some additional ideas, \cite{SarmaGP11} obtained space-efficient algorithms for estimating {\em PageRank} on graph streams. Recall that the PageRank of a webpage corresponds to the probability that a person that randomly clicks on web links arrives at this particular page\footnote{Given a web-graph $G = (V,E)$ representing the webpages and links between them, the PageRank of the vertices satisfy $\mathsf{PageRank}(u) = \sum_{(v,u) \in E} \mathsf{PageRank}(v)/d(v)$, simultaneously for all $u$, where $d(\cdot)$ denotes the out-degree, \cite{bp98}.}. However, scanning the sequence of edges $\tilde{O}(\sqrt{L})$ times may be time-inefficient in many realistic settings.

In the one-pass streaming setting, a folklore algorithm with $\tilde{O}(n \cdot L)$ space complexity for simulating $L$-step random walks is known~\cite{SarmaGP11} (see \autoref{sec:tech-algo} for a description of this algorithm), and it is proved to be optimal~\cite{Jin19}. We mention that the work of~\cite{Jin19} also considers random walks on \emph{undirected graphs}, and shows that $\tilde{\Theta}(n \cdot \sqrt{L})$ space is both necessary and sufficient for simulating $L$-step random walks on undirected graphs with $n$ vertices in one pass.

Both of these known algorithms for general directed graphs have their advantages and disadvantage (either requiring many passes or more space). A natural question is whether one can interpolate between these two results and obtain an algorithm with pass complexity much smaller than $\sqrt{L}$, yet with a space complexity much smaller than $n \cdot L$. Prior to our work, it was not even known if an $o(\sqrt{L})$-pass streaming algorithm with $n \cdot L^{0.99}$ space is possible.

\subsection{Our Results}

We answer the above question in the affirmative by giving a two-pass streaming algorithm with $\WT{O}(n \cdot \sqrt{L})$ space for sampling a random walk of length $L$ on a directed graph with $n$ vertices. We complement this result by an almost matching $\WT{\Omega}(n \cdot \sqrt{L})$ lower bound on the space complexity of every two-pass streaming algorithm for this problem. In fact, our two-pass lower bound generalizes to an $\WT{\Omega}(n \cdot L^{1/p})$ lower bound on the space consumption of any $p$-pass algorithm, for a constant $p$.

\subsubsection{Two-Pass Algorithm for Random Walk Sampling} \label{sec:intro-ub}

For a directed graph $G = (V,E)$, a vertex $\ustart \in V$ and a non-negative integer $L$, we use $\RW_L^G(\ustart)$ to denote the distribution of $L$-step random walks $(v_0,\ldots,v_L)$ in $G$ starting from $v_0=\ustart$ (see~\autoref{sec:prelim:graphs} for formal definitions). For a distribution $\calD$ over a finite domain $\Omega$, we say that a randomized algorithm {\em samples} from $\calD$ if, over its internal randomness, it outputs an element $\omega \in \Omega$ distributed according to $\calD$. We give a space-efficient streaming algorithm for (approximate) sampling from~$\RW_L^G(\ustart)$ with small error:

\begin{theorem}[Two-pass algorithm]\label{theo:algo-main}
	There exists a streaming algorithm $\algotwopass$ that given an $n$-vertex {\em directed} graph $G = (V,E)$, a starting vertex $\ustart \in V$, a non-negative integer $L$ indicating the number of steps to be taken, and an error parameter $\delta \in (0,1/n)$, satisfies the following conditions:
	\begin{enumerate}
		\item $\algotwopass$ uses at most $\tilde{O}(n \cdot \sqrt{L} \cdot \log \delta^{-1})$ space\footnote{The $\tilde{O}$ hides logarithmic factors in $n$. We may assume without loss of generality that $L \leq n^2$, as otherwise $n \cdot \sqrt{L} > n^2$ and that algorithm can store the entire input graph.} and makes two passes over the input graph~$G$.
		\item $\algotwopass$ samples from some distribution $\calD$ over $V^{L+1}$ satisfying $\tvd{\calD}{\RW^G_L(\ustart)} \le \delta$. 
	\end{enumerate}
\end{theorem}


Our algorithm can also be generalized to the \emph{turnstile} model, paying a $\poly\log n$ factor in the space usage. See \autoref{sec:turnstile}.

Observe that our algorithm $\algotwopass$  allows for a considerable saving in space compared to the folklore single-pass algorithm ($\WT{O}(n \cdot \sqrt{L})$ vs $\WT{O}(n \cdot L)$) and considerable saving in the number of passes compared to~\cite{SarmaGP11} ($2$ vs $\tilde{O}(\sqrt{L}$)), at least if we allow some small error $\delta$.

We mention that $\algotwopass$ can also be used to sample a random path from {\em every} vertex\footnote{Note, however, that the random walks from different vertices in the graph may be correlated.} of $G$ with the same storage cost of $\tilde{O}(n \cdot \sqrt{L} \cdot \log \delta^{-1})$ and two passes\footnote{We count towards the space complexity only the space on the work tape used by the algorithm and do not count space on the output tape (otherwise an $\Omega(n \cdot L)$ lower bound is trivial). }. This is because $\algotwopass$ satisfies the useful property of \emph{obliviousness to the starting vertex} $\ustart$, meaning that it scans the input graph before the start vertex is revealed. More formally, we say that an algorithm $\algo$ is oblivious to the starting vertex if it first runs a {\em preprocessing} algorithm $\prealgo$ and then a {\em sampling} algorithm $\sampalgo$; the algorithm $\prealgo$ reads the input graph stream \emph{without} knowing the starting vertex $\ustart$ (if $\algo$ is a $p$-pass streaming algorithm, $\prealgo$ makes $p$ passes over the input graph stream), and outputs a string;~$\sampalgo$ takes both the string outputted by $\prealgo$ and a starting vertex $\ustart$ as an input, and outputs a walk on the input graph $G$. 


\newpage
\subsubsection{Lower Bounds}  \label{sec:intro-lb}

We prove the following lower bound:

\newcommand{\thmstmtLB}{Fix a constant $\beta\in(0, 1]$ and an integer $p\geq 1$. Let $n \geq 1$ be a sufficiently large integer and let $L=\lceil n^{\beta}\rceil$.
Any randomized $p$-pass streaming algorithm that, given an $n$-vertex {\em directed} graph $G = (V,E)$ and a starting vertex $\ustart \in V$, samples from a distribution~$\calD$ such that $\tvd{\calD}{\RW^G_L(\ustart)} \le 1 - \frac{1}{ \log^{10} n }$
requires  $\WT{\Omega}(n \cdot L^{1/p})$ space.}
\begin{theorem}[Multi-pass lower bound]\label{theo:lowb-main-multi-pass}
\thmstmtLB
\end{theorem}


Plugging in $p=2$ in \autoref{theo:lowb-main-multi-pass}, implies that our two-pass algorithm  from \autoref{theo:algo-main} is essentially optimal. Also, with $p=1$, the theorem reproduces the one-pass lower bound by~\cite{Jin19}. In addition, \autoref{theo:lowb-main-multi-pass} rules out the possibility of a semi-streaming algorithm with any constant number of passes. 

Recall from \autoref{sec:intro-ub}, that our two-pass algorithm $\algotwopass$ utilizes $\WT{O}(n \cdot \sqrt{L})$ space and is oblivious to the starting vertex. Interestingly, we are able to show that {\em any} oblivious algorithm for random walk sampling (with any number of passes) requires $\WT{\Omega}(n \cdot \sqrt{L})$ space. Thus, any algorithm for random walk sampling with significantly less space than ours, has to be inherently different and have its storage depend on the starting vertex. Our lower bound for oblivious algorithms also implies that $\algotwopass$ gives an almost optimal algorithm for sampling a pass from every start vertex, even if any number of passes are allowed.

\begin{theorem}[Lower bound for oblivious algorithms]\label{theo:lowb-main-oblivious-algo}
	Let $n \geq 1$ be a sufficiently large integer and let $L$ denote an integer satisfying that $L \in [\log^{40} n, n]$.  Any randomized algorithm that is {\em oblivious to the start vertex} and given an $n$-vertex {\em directed} graph $G = (V,E)$ and a starting vertex $\ustart \in V$, samples from a distribution~$\calD$ such that $\tvd{\calD}{\RW^G_L(\ustart)} \le 1 - \frac{1}{ \log^{10} n }$ requires  $\WT{\Omega}(n \cdot \sqrt{L})$ space\footnote{In fact, we show that \autoref{theo:lowb-main-oblivious-algo} holds even if the preprocessing algorithm $\prealgo$ and the sampling algorithm $\sampalgo$ are allowed to use an arbitrarily large amount of memory, as long as $\prealgo$ passes a string of length at most (roughly) $n \cdot \sqrt{L}$ to $\sampalgo$.}.
\end{theorem}



\subsection{Discussions and Open Problems}

\paragraph*{Better space complexity with more passes?} Our results leave open a couple of interesting directions for future work. The most significant open question is to understand the streaming space complexity of sampling random walks with more than two passes. In particular, \autoref{theo:lowb-main-multi-pass} implies that a three-pass streaming algorithm has space complexity at least $\tilde{\Omega}(n \cdot L^{1/3})$. Can one get $\tilde{O}(n \cdot L^{1/3})$ space with three passes, or at least  $O(n \cdot L^{1/2 - \eps})$ space, for some constant $\eps > 0$? Note that, as explained in \autoref{sec:intro-lb}, such an algorithm must utilize its knowledge of the starting vertex when it reads the graph stream.

\autoref{theo:lowb-main-multi-pass} does not rule out semi-streaming $\tilde{O}(n)$ space algorithms even when $p$ is a moderately growing function of $n$ and $L$. In~\cite{SarmaGP11}, it is shown that such an $\tilde{O}(n)$ space algorithm exists with $p=\tilde{O}(\sqrt{L})$ passes. Does a semi-streaming algorithm with, say, $\poly\log(L)$ passes exist?

\paragraph*{Undirected graphs?} It would also be interesting to see what is the best two-pass streaming algorithm for simulating random walks on \emph{undirected graphs}. Specifically, is it possible to combine our algorithm with the algorithm from~\cite{Jin19} to obtain an improvement over the optimal $\tilde{O}( n \cdot \sqrt{L} )$ space complexity of a one-pass streaming algorithm for this problem? 

\paragraph*{Only outputting the end vertex?} Finally, our lower bounds only apply to the case where the algorithms need to output an entire random path $(v_0,\ldots,v_L)$. If instead only the last vertex $v_L$ in the random walk is required, can one design better two-pass algorithms or prove a non-trivial lower bound?



%% file: files/tech.tex
\section{Techniques}

\subsection{The Two-Pass Algorithm} \label{sec:tech-algo}

We next overview our two-pass algorithm from \autoref{theo:algo-main}, that simulates random walks with only $\WT{O}(n \cdot \sqrt{L})$ space.

\paragraph*{The folklore one-pass algorithm.} Before discussing our algorithm, it would be instructive to review the folklore $\tilde{O}(n \cdot L)$-space one-pass algorithm for simulating $L$-step random walks in a directed graph $G = (V,E)$ (for simplicity, we will always assume $L \le n$ in the discussions). The algorithm is quite simple: 

\begin{enumerate}
	\item For every vertex $v \in V$, sample $L$ of its outgoing neighbors \emph{with replacement} and store them in a list $\Lsave_v$ of length $L$ (that is, for each $j \in [L]$, the $j$-th element of $\Lsave_v$ is an \emph{independent uniformly random} outgoing vertex of $v$). This can be done in a single pass over input graph stream using reservoir sampling~\cite{Vitter85}.
	
	
	\item Given a starting vertex $\ustart \in V$, our random walk starts from $\ustart$ and repeats the following for $L$ steps: suppose we are currently at vertex $v$ and it is the $k$-th time we visit this vertex, then we go from $v$ to the $k$-th vertex in the list $\Lsave_v$.
\end{enumerate}

It is not hard to see that the above algorithm works: whenever we visit a vertex $v \in V$, the next element in the list $\Lsave_v$ will always be a uniformly random outgoing neighbor of $v$, \emph{conditioned on} the walk we have produced so far; and we will never run out of the available neighbors of $v$ as $|\Lsave_v| = L$.

\paragraph*{A naive attempt and the obstacles.} Since we are aiming at only using $\WT{O}(n \cdot \sqrt{L})$ space, a naive attempt to improve the above algorithm is to just sample and store $\tau = O(\sqrt{L})$ outgoing neighbors instead of $L$ neighbors, and simulate the walk starting from $\ustart$ in the same way. 
The issue here is that, during the simulation of an $L$-step walk, whenever one visits a vertex~$v$ more than~$\tau$ times, one would run out of available vertices in the list $\Lsave_v$, and the algorithm can no longer produce a legit random walk. For a simple example, imagine we have a star-like graph where $n-1$ vertices are connected to a center vertex via two-way edges. An $L$-step random walk starting at the center would require at least $\Omega(L)$ samples from the center's neighbors, and our naive algorithm completely breaks.

\paragraph*{Our approach: heavy and light vertices.} Observe, however, that in the above example of a star-like graph, we are only at risk of not storing enough random neighbors of the center node, as an $L$-step \emph{random} walk would only visit the other non-center vertices a very small number of times. Thus, the algorithm may simply record all edges from the center with only $O(n)$ space.
This observation inspires the following approach for a two-pass algorithm: 
\begin{enumerate}
\item In the first pass, we identify all the vertices that are likely to be visited many times by a random walk (starting from \emph{some} vertex). We call such vertices \emph{heavy}, while all other vertices are called {\em light}.
\item In the second pass, we record {\em all} outgoing neighbors of all heavy vertices, as well as $O(\tau)$ random outgoing neighbors with replacement of each of the light vertices.
\end{enumerate}

Observe that the obtained algorithm is indeed {\em oblivious to the starting vertex}: the two passes described above do not use the starting vertex. Still, given the set of outgoing neighbors stored by the second pass, we are able to sample a random walk from any start vertex.

\paragraph*{First pass: how do we detect heavy vertices?} The above approach requires that we detect, in a single pass, all vertices $v$ that with a decent probability (say, $1/\poly(n)$), are visited more than $O(\tau)$ times by an $L$-step random walk. 
To this end, we observe that if a random walk visits a vertex $v$ more than $\tau$ times,  this random walk must follow more than $\tau - 1$ self-circles around $v$ in $L$ steps. This, in turn, implies that a random walk that starts from $v$ is likely to return to $v$ in roughly $L/\tau = O(\sqrt{L})$ steps.

The above discussion suggests the following definition of heavy vertices: a vertex $v$ is \emph{heavy}, if a random walk starting from $v$ is likely (say, with probability at least $1/3$) to revisit~$v$ in $O(\sqrt{L})$ steps. Indeed, this property is much easier to detect: we can run $O(\log n)$ independent copies of the folklore one-pass streaming algorithms to sample $O(\log n)$ $O(\sqrt{L})$-step random walks starting from~$v$, and count how many of them return to~$v$ at some step. 

\paragraph*{Second pass: can we afford to store the neighbors?}
In~\autoref{lemma:light-is-not-visited-much}, we show that for a light vertex $v$, an $L$-step random walk starting at any vertex visits $v$ $O(\sqrt{L})$ times with high probability. Therefore, in the second pass, we can safely record only $O(\sqrt{L})$ outgoing neighbors for all light vertices. Still, we have to record all the outgoing neighbors for heavy vertices. 

The crux of our analysis is a {\em structural result} about directed graphs, showing that the total outgoing degree of all heavy vertices is bounded by $O(n \cdot \sqrt{L})$, and therefore we can simply store all of their outgoing neighbors. This is proved in \autoref{lemma:heavy-total-out-degree-bounded}, which may also be of independent interest. 

\paragraph*{Intuition behind the structure lemma.} Finally, we discuss the insights behind the above structure lemma for directed graphs. We will use $\dout(v)$ to denote the number of outgoing neighbors of $v$. For concreteness, we now say a vertex $v$ is heavy if a random walk starting from $v$ revisits $v$ in $\sqrt{L}$ steps with probability at least $1/3$. 

Let $\Vheavy \subseteq V$ be the set of heavy vertices and let $v \in \Vheavy$. By a simple calculation, one can see that for at least a $1/6$ fraction of outgoing neighbors $u$ of $v$, a random walk starting from $u$ visits $v$ in $\sqrt{L}$ steps with probability at least $1/6$. The key insight is to consider the number of pairs $(u,v) \in V^2$ such that a random walk starting from $u$ visits $v$ in $\sqrt{L}$ steps with probability at least $1/6$. We will use $\calS$ to denote this set.

\begin{itemize}
	\item By the previous discussions, we can see that for each heavy vertex $v$, it adds at least $1/6$ $\dout(v)$ pairs to the set $\calS$. Hence, we have 
	\begin{align}\label{eq:lower_bound_on_calS}
		|\calS| \ge \frac{1}{6} \cdot \sum_{v \in \Vheavy} \dout(v).
	\end{align}
	
	
	\item On the other hand, it is not hard to see that for each vertex $v$, there are at most $O(\sqrt{L})$ many pairs of the form $(v,u) \in \calS$, since a $\sqrt{L}$-step walk can visit only $\sqrt{L}$ vertices. So we also have 
	\begin{align}\label{eq:upper_bound_on_calS}
		|\calS| \le O(n \sqrt{L}).
	\end{align} 
\end{itemize}

Putting the above (\autoref{eq:lower_bound_on_calS} and \autoref{eq:upper_bound_on_calS}) together, we get the desired bound 
\begin{align*}
\sum_{v \in \Vheavy} \dout(v) \le O(n \sqrt{L}).
\end{align*}

\subsection{Lower Bound for $p$-Pass Algorithms}

We now describe the ideas behind the proof of \autoref{theo:lowb-main-multi-pass}, our $\WT{\Omega}(n \cdot L^{1/p})$ space lower bound for $p$-pass randomized streaming algorithms for sampling random walks.  We mention that many of the tools developed for proving space lower bounds are not directly applicable when one wishes to lower bound the space complexity of a {\em sampling} task and are more suitable for proving lower bounds on the space required to compute a function or a search problem\footnote{One such tool that cannot be used directly for our purpose is the very useful {\em Yao's minimax principle} \cite{y77} that allows proving randomized communication lower bounds by proving the corresponding distributional (deterministic) communication lower bounds.}.

\paragraph*{From sampling to function computation.} Our way around this is to first prove a reduction from streaming algorithms that sample a random walk from $\ustart$ to streaming algorithms that compute the $(p+1)$-neighborhood of the vertex $\ustart$. This is done by considering a graph where a random walk returns to the vertex $\ustart$ every $p+2$ steps. If $p$ is a constant, then a random walk of length~$L$ on such a graph can be seen as $L/(p+2)=O(L)$ copies of a random walk of length $p+2$. Observe that if the $(p+1)$-neighborhood of the vertex $\ustart$ has (almost) $L$ vertices (and the probability of visiting each vertex is more or less uniform), then a random walk of length $L$ is likely to visit all the vertices in the neighborhood and an algorithm that samples a random walk also outputs the entire neighborhood with high probability. 

\paragraph*{A lower bound for computing the $(p+1)$-neighborhood via pointer-chasing.} Having reduced sampling a random walk to outputting the $(p+1)$-neighborhood, we now need to prove that a space efficient $p$-pass streaming algorithms cannot output the $(p+1)$-neighborhood of $\ustart$, if this neighborhood has roughly $L$ vertices. This is reminiscent of the ``{\em pointer-chasing}'' lower bounds found in the literature.

Pointer-chasing results are typically concerned with a graph with $p+1$ layers of vertices ($p$ layers of edges) and show that given a vertex in the first layer, finding a vertex that is reachable from it in the last layer cannot be done with less than $p$ passes, unless the memory is huge. Classical pointer-chasing lower bounds ({\em e.g.}, \cite{NW91}), consider graphs where the out-degree of each vertex is $1$, thus the start vertex reaches a unique vertex in the last layer. Unfortunately, this type of pointer-chasing instances are very {\em sparse} and a streaming algorithm can simply remember the entire graph in one pass using $\tilde{O}(n)$ memory.


Since we wish to have roughly $L$ vertices in a  $(p+1)$-neighborhood of $\ustart$, the out-degree of each vertex should be roughly $\Omega(L^{\frac{1}{p+1}})$ (assuming uniform degrees). Pointer-chasing lower bounds for this type of {\em dense} graphs were also proved ({\em e.g.}, \cite{go16} and \cite{fkms+09}), showing that $p$-pass algorithms essentially need to store an entire layer of edges, which is $\Omega(n \cdot L^{\frac{1}{p+1}})$ in our case. However, this still does not give us the $\Omega(n \cdot L^{\frac{1}{p}})$ lower bound we aspire for (and which is tight, at least for two passes). 


\paragraph*{Towards a tight lower bound: combining dense and sparse.}  To get a better lower bound, we construct a hard instance that is a combination of the two above mentioned types of pointer-chasing instances, the dense and the sparse. Specifically, for a $p$-pass lower bound, we construct a layered graph with $p+2$ layers of vertices $V_1,\ldots,V_{p+2}$, where the first layer has only one vertex $\ustart$ and all the other layers are of equal size (see \autoref{fig:graph}). To ensure that vertex $\ustart$ is reached every $p+2$ steps, we connect all vertices in the last layer to $\ustart$. Every vertex in layers $V_2,\ldots,V_{p+1}$ connects to a random set of roughly $L^{\frac{1}{p}}$ vertices in the next layer. Using Guruswami and Onak style arguments (\cite{go16}), it can be shown that when the edges are presented to the algorithm from right to left, finding a vertex in layer $V_{p+2}$ that is reachable form a given vertex in $V_2$ with a ($p-1$)-pass algorithm requires $\Omega(n \cdot L^{\frac{1}{p}})$ space. We ``squeeze out'' an extra pass in the algorithm by connecting the start vertex $\ustart$ in $V_1$ to a single random vertex in $V_2$. Note that with this construction, it is indeed the case that a $(p+1)$-neighborhood of $\ustart$ consists of only roughly $L$ vertices, but still, the out-degrees of vertices in $V_2,\ldots,V_{d+1}$ are roughly $L^{\frac{1}{p}}$ instead of only $L^{\frac{1}{p+1}}$.

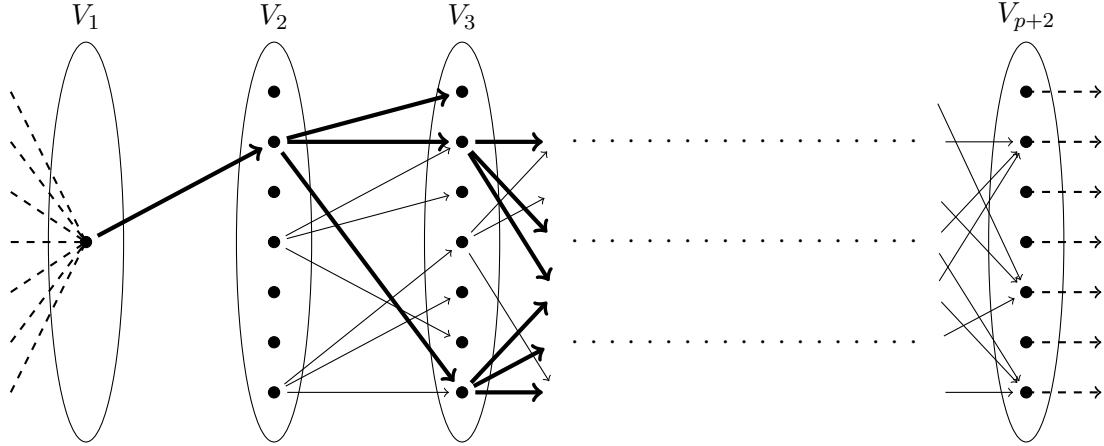
\begin{figure}
	\centering
	\begin{tikzpicture}[vrtx/.style = {draw, fill, circle, inner sep = 1.5pt} ] 
	\node[vrtx] (s00) at (0,0) {};
	\foreach \x in {1,2,5} {
		\foreach \y in {-3, -2, ..., 3} {
			\node[vrtx] (s\x\y) at (\x*2.5, \y/1.5) {};
		}
	}
	\node (n0) at (0,3) {$V_1$};
	\foreach \x/\n in {1/2,2/3,5/p+2} {
		\node (n\x) at (\x*2.5, 3) {$V_{\n}$};
	}
	\draw (0, 0) ellipse (0.5cm and 2.66cm);
	\foreach \x in {1,2,5} {
		\draw (\x*2.5, 0) ellipse (0.5cm and 2.66cm);
	}

	\foreach \x in {5} {
		\foreach \y in {-3, -2, ..., 3} {
			\draw[thick, dashed, ->] (\x*2.5, \y/1.5) -- (\x*2.5 + 1, \y/1.5) ;
			\draw[thick, dashed] (-1, \y/1.5) -- (0,0) ;
		}
	}

	\foreach \x in {1, ..., 19} {
		\foreach \y in {-2,0,2} {
			\node (s\x\y) at (\x/4 + 6.25, \y/1.5) {$\cdot$};
		}
	}

	\draw[ultra thick, ->, shorten <=5pt, shorten >= 5pt] (0, 0) -- (2.5, 4/3);

	\foreach \x/\y/\z in {1/2/-3, 1/2/2, 1/2/3} {	
		\draw[ultra thick, ->, shorten <=5pt, shorten >= 5pt] (\x*2.5, \y/1.5) -- (\x*2.5 + 2.5, \z/1.5) ;
	}
	
	\foreach \x/\y/\z in {1/-3/-3, 1/-3/-1, 1/-3/0, 1/0/-2, 1/0/2, 1/0/1} {	
		\draw[->, shorten <=5pt, shorten >= 5pt] (\x*2.5, \y/1.5) -- (\x*2.5 + 2.5, \z/1.5) ;
	}
	
	\foreach \x/\y/\z in {2/2/2, 2/2/0, 2/2/-1, 2/-3/-3, 2/-3/-2, 2/-3/-1} {	
		\draw[ultra thick, ->, shorten <=5pt, shorten >= 5pt] (\x*2.5, \y/1.5) -- (\x*2.5 + 1.25, \z/1.5) ;
	}
	
	\foreach \x/\y/\z in {2/0/1, 2/0/2, 2/0/-3} {	
		\draw[->, shorten <=5pt, shorten >= 5pt] (\x*2.5, \y/1.5) -- (\x*2.5 + 1.25, \z/1.5) ;
	}
	
	\foreach \x/\y/\z in {5/2/2, 5/2/0, 5/2/-1, 5/-1/-2, 5/-1/3, 5/-1/1, 5/-3/0, 5/-3/-3, 5/-3/-1} {	
		\draw[ ->, shorten <=5pt, shorten >= 5pt] (\x*2.5 - 1.25, \z/1.5) -- (\x*2.5, \y/1.5) ;
	}

	\end{tikzpicture}
	\caption{A depiction of our hard instance for $p$-pass streaming algorithms. Some edges omitted.}
	\label{fig:graph}
\end{figure}


\subsection{Lower bounds for Oblivious Algorithms} Finally, we discuss the intuitions behind the proof of~\autoref{theo:lowb-main-oblivious-algo}, showing that any algorithm that is oblivious to the starting vertex must use $\WT{\Omega}(n \cdot \sqrt{L})$ space space. Our proof is based on a reduction from a multi-output generalization of the well-studied $\INDEX$ problem for one-way communication protocols, denoted by $\INDEX_{m,\ell}$. In $\INDEX_{m,\ell}$, Alice gets $\ell$ strings $X_1,\dotsc,X_\ell \in \{0,1\}^m$ and Bob gets an index $i \in [\ell]$. Alice sends a message to Bob and then Bob is required to output the string $X_i$. (Note that when $m = 1$ it becomes the original $\INDEX$ problem). 

It is not hard to show that any one-way communication protocol solving $\INDEX_{m,\ell}$ with non-trivial probability (say, $1/\poly\log(m)$) requires Alice to send at least $\widetilde{\Omega}(m\ell)$ bits to Bob (see~\autoref{lemma:lowb-INDEX}).

Our key observation here is that if there is a starting vertex oblivious algorithm $\algo = (\prealgo,\sampalgo)$ with $S$ space for approximate simulation of an $L = \WT{O}(m)$-step random walk on a graph with $n = O(\sqrt{m} \cdot \ell)$ vertices, then it implies a one-way communication protocol for $\INDEX_{m,\ell}$ with communication complexity $S$ and a decent success probability. Recall the lower bound for $\INDEX_{m,\ell}$, we immediately have $S = \widetilde{\Omega}(m\ell) = \widetilde{\Omega}(n\sqrt{L})$.

In more detail, given an $m$-bit string $X$, we will build an $O(\sqrt{m})$-vertex graph $H(X)$ by encoding all bits of $X$ as \emph{existence/non-existence} of edges in $H$ (this is possible since there are more than $m$ potential edges in $H$). We also add some artificial edges to $H$ to make sure it is strongly connected. Our construction will make sure that an $L = \WT{O}(m)$ steps random walk in $H$ will reveal all edges in $H$ with high probability, which in turn reveals all bits of $X$ (see the proof of~\autoref{theo:lowb-main-oblivious-algo} for more details).

Now the reduction can be implemented as follows: given $\ell$ strings $X_1,\dotsc,X_\ell \in \{0,1\}^m$, Alice constructs a graph $G = \bigsqcup_{i=1}^{\ell} H(X_i)$, as the joint union of $\ell$ graphs. Note that $G$ has $n = O(\sqrt{m} \cdot \ell)$ vertices. Alice then runs the preprocessing algorithm $\prealgo$ on $G$ to obtain a string $M$, and sends it to Bob. Given an index $i \in [\ell]$, Bob simply runs $\sampalgo$ with $M$ together with a suitable starting vertex inside the $H(X_i)$ component of $G$. By previous discussions, this reveals the string $X_i$ with high probability and proves the correctness of this reduction. Hence, the space complexity of $\algo$ must be $\widetilde{\Omega}(m\ell) = \widetilde{\Omega}(n\sqrt{L})$.

\subsection*{Organization of this paper}
In~\autoref{sec:prelim} we introduce the necessary preliminaries for this paper. In~\autoref{sec:two-pass-algo} we present our nearly optimal two-pass streaming algorithm for simulating random walks and prove \autoref{theo:algo-main}. In~\autoref{sec:lb} we prove our lower bounds against general multi-pass streaming algorithms for simulating random walks (\autoref{theo:lowb-main-multi-pass}). In~\autoref{sec:info} we present some additional preliminaries in information theory. In~\autoref{sec:missing-proofs} we provide some missing proofs in~\autoref{sec:lb}. In~\autoref{sec:lowb-input-oblivious} we prove~\autoref{theo:lowb-main-oblivious-algo}.

%% file: files/prelim.tex

\section{Preliminaries}
\label{sec:prelim}

\subsection{Notation}

Let $n \in \mathbb{N}$. We use $[n]$ to denote the set $\{1,\dotsc,n\}$. We often use sans-serif letters (\eg, $\sfX$) to denote random variables, and calligraphic font letters (\eg, $\calX$) to denote distributions. For two random variables $\sfX$ and $\sfY$, and for $Y \in \supp(\sfY)$, we use $(X|\sfY = Y)$ to denote $\sfX$ conditioned on $\sfY = Y$. For two lists $a$ and $b$, we use $a \circ b$ to denote their concatenation. 

For two distributions $\calD_1$ and $\calD_2$ on set $\calX$ and $\calY$ respectively, we use $\calD_1 \otimes \calD_2$ to denote their product distribution over $\calX \times \calY$, and $\tvd{\calD_1}{\calD_2}$ to denote the total variation distance between them.

\subsection{Graphs}\label{sec:prelim:graphs}
In this paper we will always consider directed graphs without multi-edges. A directed $G$ is a pair $(V,E)$, where $V$ is the vertex set and $E \subseteq V \times V$ is the set of all edges.

For a vertex $u$ in a graph $G = (V,E)$,  we let $\Nout^G(u) \coloneqq \{ v : (u,v) \in E \}$ and $\Nin^G(u)  \coloneqq \{ v : (v,u) \in E \}$. We also use $\dout^G(u)$ and $\din^G(u)$ to denote its out and in degrees (\ie, $|\Nout^G(u)|$ and $|\Nin^G(u)|$). For an edge $(u,v) \in E$, we say $v$ is the \emph{out-neighbor} of $u$ and $u$ is the \emph{in-neighbor} of $v$.

\paragraph*{Random walks on directed graphs.} For a vertex $u$ in a graph $G = (V,E)$ and an non-negative integer $L$, an $L$-step random walk $(v_0,v_1,\dotsc,v_L)$ starting at $u$ is generated as follows: set $v_0 = u$, for each $i \in [L]$, we draw $v_{i}$ uniformly random from $\Nout^G(u)$.  We say that $v_0 = u$ is the $0$-th vertex on the walk, and $v_i$ is the $i$-th vertex for each $i \in [L]$. We use $\RW^G_L(u)$ to denote the distribution of an $L$-step random walk starting from $u$ in $G$.

We use $\visit^G_{[a,b]}(u,v)$ to denote the probability of a $b$-step random walk starting from $u$ visits $v$ between the $a$-th vertex and $b$-th vertex on the walk.

We often omit the superscript $G$ when the graph $G$ is clear from the context.

\paragraph*{Starting vertex oblivious algorithms.} Now we formally define a starting vertex oblivious streaming algorithm for simulating random walks.

\begin{definition}\label{defi:starting-vertex-oblivious}
	We say a $p$-pass $S$-space streaming algorithm $\algo$ for simulating random walks is starting vertex oblivious, if $\algo$ can be decomposed into a preprocessing subroutine $\prealgo$ and a sampling subroutine $\sampalgo$, such that: 
	\begin{enumerate}
		\item (\textbf{Starting vertex oblivious preprocessing phase}) $\prealgo$ makes $p$ passes over the input graph stream, using at most $S$ words of space. After that, $\prealgo$ outputs at most $S$ words, denoted as $M$.
		\item (\textbf{Sampling phase}) $\sampalgo$ takes both the starting vertex $\ustart$ and $M$ as input, and outputs a desired walk starting from $\ustart$, using at most $S$ words of space.
	\end{enumerate}
\end{definition}

\subsection{Useful Concentration Bounds on Random Variables}

The following standard concentration bounds will be useful for us.

\begin{lemma}[Multiplicative Chernoff bound, \cite{c52}]
	\label{lemma:chernoff}
	Suppose $X_1, \cdots, X_n$ are independent random variables taking values in $[0,1]$. Let $X$ denote their sum and let $\mu = \E[X]$ denote the sum's expected value. Then,
	\begin{alignat*}{2}
	\Pr\left(X\geq (1+\delta )\mu\right)&\leq \e^{-{\frac {\delta ^{2}\mu }{2 + \delta}}},\hspace{1cm} &&\forall 0 \leq \delta ,\\
	\Pr\left(X\leq (1-\delta )\mu\right)&\leq \e^{-{\frac {\delta^{2}\mu }{2}}}, &&\forall 0 \leq \delta \leq 1.
	\end{alignat*}
	In particular, we have that:
	\begin{alignat*}{2}
	\Pr\left(X\geq (1+\delta )\mu\right)&\leq \e^{-\frac{\delta \mu }{3} \cdot \min(\delta, 1) },\hspace{1cm} &&\forall 0 \leq \delta,\\
	\Pr\left( \lvert{X  - \mu }\rvert \geq \delta \mu  \right)&\leq 2 \cdot \e^{-{\frac {\delta ^{2}\mu }{3}}}, &&\forall 0 \leq \delta \leq 1.
	\end{alignat*}
\end{lemma}

We also need the following Azuma-Hoeffding inequality.

\begin{lemma}[Azuma-Hoeffding inequality, \cite{azuma1967weighted,hoeffding1994probability}]\label{lemma:azuma-hoeffding}
	Let $Z_0,\dotsc,Z_n$ be random variables satisfying (1) $\E[|Z_i|] < \infty$ for every $i \in \{0,\dotsc,n\}$ and $\E[Z_i | Z_0,\dotsc,Z_{i-1}] \le Z_{i-1}$ for every $i \in [n]$ (\ie, $\{Z_i\}$ forms a supermartingale) and (2) for every $i \in [n]$, $|Z_i - Z_{i-1}| \le 1$, then for all $\lambda > 0$, we have
	\[
	\Pr[Z_n - Z_0 \ge \lambda] \le \exp(-\lambda^2/2n).
	\]
	
\end{lemma}

In particular, the following corollary will be useful for us.

\begin{corollary}[Azuma-Hoeffding inequality for Boolean random variables, \cite{azuma1967weighted,hoeffding1994probability}]\label{cor:azuma-hoeffding}
	Let $X_1,\dotsc,X_n$ be random variables satisfying $X_i \in \{0,1\}$ for each $i \in [n]$. Suppose that $\E[X_i | X_1,\dotsc,X_{i-1}] \le p_i$ for all $i$. Then for any $\lambda > 0$,
	\[
	\Pr\left[ \sum_{i=1}^{n} X_i \ge \lambda + \sum_{i=1}p_i \right] \le \exp(-\lambda^2/2n).
	\]
\end{corollary}
\begin{proof}
	For $i \in \{0,\dotsc,n\}$, let $Z_i = \sum_{j=1}^{i} (X_j - p_j)$. From the assumption one can see that all the $Z_i$ form a supermartingale and $|Z_i - Z_{i-1}| \le 1$, hence the corollary follows directly from~\autoref{lemma:azuma-hoeffding}.
\end{proof}

\subsection{Standard Lemmas} \label{sec:concineq}

We will need (a weak form of) Stirling's approximation. We include a proof for completeness.
\begin{lemma}
	\label{lemma:stirling} For all $n > 0$, we have
	\[
	\e \cdot \paren*{ \frac{n}{\e} }^n \leq n! \leq \e n \cdot \paren*{ \frac{n}{\e} }^n .
	\]
\end{lemma}
\begin{proof}
	We have:
	\[
	\frac{ n^n }{ n! } = \frac{ n^{n-1} }{ (n-1)! } = \prod_{i = 1}^{n-1} \paren*{ \frac{i+1}{i} }^i = \prod_{i = 1}^{n-1} \paren*{ 1 + \frac{1}{i} }^i \leq \prod_{i = 1}^{n-1} \e = \e^{n-1} .
	\]
	We also have:
	\[
	\frac{ n! }{ n^{n+1} } = \frac{ (n-1)! }{ n^n } = \prod_{i = 1}^{n-1} \paren*{ \frac{i}{i+1} }^{i+1} \leq \prod_{i = 1}^{n-1} \paren*{ 1 - \frac{1}{i+1} }^{i+1} \leq \prod_{i = 1}^{n-1} \e^{-1} = \e^{1-n} .
	\]
	Rearranging gives the result.
\end{proof}

The following bound on binomial coefficients follows:
\begin{lemma}
	\label{lemma:binomial} For all $0 < k < n$, we have
	\[
	\frac{1}{10n^2} \cdot \paren*{ \frac{n}{k} }^k \cdot \paren*{ \frac{n}{n-k} }^{n-k} \leq \binom{n}{k} \leq n \cdot \paren*{ \frac{n}{k} }^k \cdot \paren*{ \frac{n}{n-k} }^{n-k} .
	\]
\end{lemma}
\begin{proof}
	We have:
	\begin{align*}
	\binom{n}{k} &= \frac{ n! }{ (n-k)! \cdot k! } \\
	&\geq \frac{ \e \cdot \paren*{ \frac{n}{\e} }^n }{ \e k \cdot \paren*{ \frac{k}{\e} }^k \cdot \e (n-k) \cdot \paren*{ \frac{n-k}{\e} }^{n-k} } \tag{\autoref{lemma:stirling}} \\
	&\geq \frac{1}{10n^2} \cdot \paren*{ \frac{n}{k} }^k \cdot \paren*{ \frac{n}{n-k} }^{n-k} .
	\end{align*}
	We also have:
	\begin{align*}
	\binom{n}{k} &= \frac{ n! }{ (n-k)! \cdot k! } \\
	&\leq \frac{ \e n \cdot \paren*{ \frac{n}{\e} }^n }{ \e \cdot \paren*{ \frac{k}{\e} }^k \cdot \e \cdot \paren*{ \frac{n-k}{\e} }^{n-k} } \tag{\autoref{lemma:stirling}} \\
	&\leq n \cdot \paren*{ \frac{n}{k} }^k \cdot \paren*{ \frac{n}{n-k} }^{n-k} .
	\end{align*}
\end{proof}

%% file: files/algo.tex
\section{Two-Pass Streaming Algorithms for Simulating Directed Random Walk}\label{sec:two-pass-algo}


In this section, we present our two-pass streaming algorithms for simulating random walks on directed graphs.

\subsection{Heavy and Light Vertices}

We first define the notion of heavy and light vertices.

\begin{definition}[Heavy and light vertices]
	Given a directed graph $G = (V,E)$ with $n$ vertices and $\ell \in \N$.
	\begin{itemize}
		\item (\textbf{Heavy vertices}.) We say a vertex $u$ is $\ell$-heavy in $G$, if $\visit_{[1,\ell]}(u,u) \ge 1/3$ (\ie, if a random walk starting from $u$ will revisit $u$ in at most $\ell$ steps with probability at least $1/3$.)
		
		\item (\textbf{Light vertices}.) We say a vertex $u$ is $\ell$-light in $G$, if $\visit_{[1,\ell]}(u,u) \le 2/3$ (\ie, if a random walk starting from $u$ will revisit $u$ in at most $\ell$ steps with probability at most $2/3$.)
	\end{itemize}

	We also let $\Vheavy^\ell(G)$ and $\Vlight^\ell(G)$ be the sets of $\ell$-heavy and $\ell$-light vertices in $G$. When $G$ and $\ell$ are clear from the context, we simply refer to them as $\Vheavy$ and $\Vlight$.
\end{definition}

\begin{remark}
	Note that if the revisiting probability is between $[1/3,2/3]$, then the vertex is considered to be both heavy and light.
\end{remark}

The following lemma is crucial for the analysis of our algorithm.
\begin{lemma}[Upper bounds on the total out-degrees of heavy vertices]\label{lemma:heavy-total-out-degree-bounded}
	Given a directed graph $G$ with $n$ vertices and $\ell \in \N$, it holds that 
	\[
	\sum_{u \in \Vheavy^\ell(G)} \dout(u) \le O(n \cdot \ell).
	\]
\end{lemma}
\begin{proof}
	We define a set $\calS$ of pairs of vertices as follows:
	\[
	\calS \coloneqq \{ (u,v) \in V^2 : \visit_{[0,\ell]}(u,v) \ge 1/6 \}.
	\]
	That is, a pair of vertices $u$ and $v$ belongs to $\calS$ if and only if an $\ell$-step random walk starting from $u$ visits $v$ with probability at least $1/6$.
	
	For each fixed vertex $u$, we further define 
	\[
	\calS_u \coloneqq \set*{v \in \Nout(u) \mid \visit_{[0,\ell]}(v,u) \ge 1/6},
	\]
	and
	\[
	\calH_u \coloneqq \set*{v \in V \mid \visit_{[0,\ell]}(u,v) \ge 1/6}.
	\]
	
	The following claim will be useful for the proof.
	\begin{claim}\label{claim:properties}
		The following two statements hold:
		\begin{enumerate}
			\item For every $u \in V$, it holds that $|\calH_u| \le O(\ell)$.
			\item For every $u \in \Vheavy$, it holds that $|\calS_u| \ge 1/6 \cdot \dout(u)$.
		\end{enumerate}
	\end{claim}
	\begin{proof}
	Fixing $u \in V$, the first item follows from the simple fact that
	\[
	\sum_{v \in V} \visit_{[0,\ell]}(u,v) \le \ell + 1.
	\]
	
	Now we move to the second item, and	fix $u \in \Vheavy$. For the sake of contradiction, suppose that $\left| \calS_u \right| < 1/6 \cdot \dout(u)$. We have
	\begin{align*}
		\visit_{[1,\ell]}(u,u) &= \E_{v \in \Nout(u)} [\visit_{[0,\ell-1]}(v,u)] \\
		&\le \E_{v \in \Nout(u)} [\visit_{[0,\ell]}(v,u)] \\
		&< \Pr_{v \in \Nout(u)}[v \in \calS_u] \cdot 1 + \Pr_{v \in \Nout(u)}[v \notin \calS_u] \cdot 1/6 < 1/6 + 1/6 < 1/3,
	\end{align*}
	a contradiction to the assumption that $u$ is heavy.
	\end{proof}
	
	Finally, note that by definition of $\calH_u$ and $\calS_u$ we immediately have
	\[
	|\calS| = \sum_{u \in V} |\calH_u| \ge \sum_{u \in V} |\calS_u|.
	\]
	
	By~\autoref{claim:properties}, we have
	\[
	\sum_{u \in \Vheavy} \dout(u) \le 6 \cdot \sum_{u \in V} |\calS_u| \le 6 \cdot \sum_{u \in V} |\calH_u| \le O\left(n \cdot \ell\right),
	\]
	which completes the proof.

\end{proof}

\subsection{A Simple One-Pass Algorithm for Simulating Random Walks}
We first describe a simple one-pass algorithm for simulating random walks, which will be used as a sub-routine in our two-pass algorithm. Moreover, this one-pass algorithm is starting vertex oblivious, which will be crucial for us later.

\paragraph*{Reservoir sampling in one pass.} Before describing our one-pass subroutine, we need the following basic reservoir sampling algorithm. 

\begin{lemma}[\cite{Vitter85}]
	Given input access to a stream of $n$ items such that each item can be described by $O(1)$ words, we can uniformly sample $m$ of them without replacement using $O(m)$ words of space.
\end{lemma}

Using $m$ independent reservoir samplers each with capacity $1$, one can also sample $m$ items from the stream \emph{with replacement} in a space-efficient way.
\begin{corollary}\label{cor:reservoir-with-replacement}
	Given input access to a stream of $n$ items such that each item can be described by $O(1)$ words, we can uniformly sample $m$ of them with replacement using $O(m)$ words of space.
\end{corollary}

\paragraph*{Description of the one-pass algorithm.} Now we describe our one-pass algorithm for simulating random walks. Our algorithm $\algosub$ is starting vertex oblivious, and can be described by a preprocessing subroutine $\prealgosub$ and a sampling subroutine $\sampalgosub$. Recall that as defined in~\autoref{defi:starting-vertex-oblivious}, $\prealgosub$ takes a single pass over the input graph streaming without knowing the starting vertex $\ustart$, and $\sampalgosub$ takes the output of $\prealgosub$ together with $\ustart$, and outputs a desired sample fo the random walk.


\begin{algorithm}[H]
	\caption{Preprocessing phase of $\algosub$: $\prealgosub(G,\tau,\Vfull)$}
	\label{algo:one-pass-preprocess}
	
	\begin{algorithmic}[1]
		\renewcommand{\algorithmicrequire}{\textbf{Input:}}
		\renewcommand{\algorithmicensure}{\textbf{Output:}}
		
		\Require One pass streaming access to a directed graph $G = (V,E)$. A parameter $\tau \in \N$. A subset $\Vfull \subseteq V$, and we also let $\Vsamp = V \setminus \Vfull$.\smallskip
		
		\State For each vertex $v \in \Vfull$, we record all its out-neighbors in the list $\Lsave_v$. (That is, $\Vfull$ stands for the set of vertices that we keep all its edges.) \smallskip
		
		\State For each vertex $v \in \Vsamp$, using~\autoref{cor:reservoir-with-replacement}, we sample $\tau$ of its out-neighbors uniformly at random with replacement in the list $\Lsave_v$. (That is, $\Vsamp$ stands for the set of vertices that we sample some of its edges.) \smallskip
		
		\State For a big enough constant $c_2 > 1$, whenever the number of out-neighbors stored exceeds $c_2 \cdot \tau \cdot n$, the algorithm stops recording them. If this happens, we say the algorithm \emph{operates incorrectly} and otherwise we say it \emph{operates correctly}. \smallskip
		
		\Ensure A collection of lists $\vecLsave = \{ \Lsave_v \}_{v \in V}$.
	\end{algorithmic}
	
\end{algorithm}
\begin{algorithm}[H]
	\caption{Sampling phase of $\algosub$: $\sampalgosub(V,\ustart,L,\Vfull, \vecLsave = \{ \Lsave_v \}_{v \in V})$}
	\label{algo:one-pass-sampling}
	
	\begin{algorithmic}[1]
		\renewcommand{\algorithmicrequire}{\textbf{Input:}}
		\renewcommand{\algorithmicensure}{\textbf{Output:}}
		
		\Require A starting vertex $\ustart$. The path length $L \in \N$. A subset $\Vfull \subseteq V$, and we also let $\Vsamp = V \setminus \Vfull$.\smallskip
		
		\State Let $v_0 = \ustart$. For each $v \in V$, we set $k_v = 1$.\smallskip
		
		\For{ $i \coloneqq 1 \to L$} 
		\If{$v_{i-1} \in \Vfull$}
		\State $v_i$ is set to be a uniformly random element from $\Lsave_{v_{i-1}}$
		\ElsIf{$k_{v_{i-1}} > |\Lsave_{v_{i-1}}|$} 
		\State
		\Return failure
		\Else 
		\State $v_i \leftarrow (\Lsave_{v_{i-1}})_{k_{v_{i-1}}}$.
		\State $k_{v_{i-1}} \leftarrow k_{v_{i-1}} + 1$. 
		\EndIf
		\EndFor\smallskip
		
		\Ensure The walk $(v_0,v_1,\dotsc,v_L)$.
	\end{algorithmic}
	
\end{algorithm}

\paragraph*{Analysis of the one-pass algorithm.} Now we analyze the correctness of our one-pass algorithm. We first observe its space complexity can be easily bounded.

\begin{observation}[Space complexity of $\algosub$]\label{ob:sub-procedure-space}
	Given a directed graph $G = (V,E)$ with $n$ vertices. For every $\tau \in \N$ and subset $\Vfull \subseteq V$, $\prealgosub(G,\tau,\Vfull)$ always takes at most $O(\tau \cdot n)$ words of space.
\end{observation}

Next we bound the statistical distance between its output distribution and the correct distribution of the random walk by the following lemma.

\begin{lemma}[Correctness of $\algosub$]\label{lemma:sub-procedure-correctness}
	Given a directed graph $G = (V,E)$ with $n$ vertices. For every integers $\tau,L \in \N$ and subset $\Vfull \subseteq V$ such that $\tau \cdot (n - |\Vfull|) + \sum_{v \in \Vfull} \dout(v) \le c_2 \cdot \tau \cdot n$, let $\bvecLsave$ be random variable of the output of $\prealgosub(G,\tau,\Vfull)$. For every $\ustart \in V$, the output distribution of $\sampalgosub(V,\ustart,L,\Vfull, \bvecLsave)$ has statistical distance $\beta$ to $\RW^{G}_L(\ustart)$, where $\beta$ is the probability that $\sampalgosub(V,\ustart,L,\Vfull, \bvecLsave)$ outputs failure. 
\end{lemma}
\begin{proof}
	Conclude from $\tau \cdot (n - |\Vfull|) + \sum_{v \in \Vfull} \dout(v) \le c_2 \cdot \tau \cdot n$ that $\prealgosub(G,\tau,\Vfull)$ always operates correctly.
		
	To bound the statistical distance between the distribution of $\sampalgosub(V,\ustart,L,\Vfull, \bvecLsave)$ and $\RW^{G}_L(\ustart)$. We construct another random variable $(\bvecLsave)'$, in which for every vertex $u$, we sample another $L$ out-neighbors of $u$ uniformly at random with replacement, and add them to the end of the list $\Lsave_v$ in $\bvecLsave$.
	
	Note that $\sampalgosub(V,\ustart,L,\Vfull, (\bvecLsave)')$ never outputs failure, and distributes exactly the same as $\RW^{G}_L(\ustart)$. On the other hand, $\sampalgosub(V,\ustart,L,\Vfull, (\bvecLsave)')$ and $\sampalgosub(V,\ustart,L,\Vfull, \bvecLsave)$ are the same as long as $\sampalgosub(V,\ustart,L,\Vfull, \bvecLsave)$ does not output failure, which completes the proof.
\end{proof}

The following corollary follows immediately from the lemma above. (Note that this special case exactly corresponds to the folklore one-pass streaming algorithm for simulating random walks.)

\begin{corollary}\label{cor:sub-procedure-correctness}
	Given a directed graph $G = (V,E)$ with $n$ vertices and an integer $L \in \N$. Let $\bvecLsave$ be random variable of the output of $\prealgosub(G,L,\emptyset)$. For every $\ustart \in V$, the output distribution of $\sampalgosub(V,\ustart,L,\emptyset, \bvecLsave)$ distributes identically as $\RW^{G}_L(\ustart)$.
\end{corollary}

\subsection{Two-Pass Streaming Algorithm for Simulating Random Walks}

\paragraph*{Description of the two-pass algorithm.} Now we are ready to describe our two pass algorithm $\algotwopass$, which is also starting vertex oblivious, and can be described by the following two sub-routines $\prealgomain$ and $\sampalgomain$.

\begin{algorithm}[H]
	\caption{Preprocessing phase of $\algotwopass$: $\prealgomain(G,L,\delta)$}
	\label{algo:two-pass-preprocess}
	\begin{algorithmic}[1]
		\renewcommand{\algorithmicrequire}{\textbf{Input:}}
		\renewcommand{\algorithmicensure}{\textbf{Output:}}
		
		\Require A directed graph $G = (V,E)$ with $n$ vertices. An integer $L \in \N$. A failure parameter $\delta \in (0,1/n)$. We also let $\ell = \sqrt{L}$, and $\gamma = c_1 \cdot \log \delta^{-1}$ where $c_1 \ge 1$ is a sufficiently large constant to be specified later. 
		\smallskip
		
		\State \textbf{First pass: estimation of heavy and light vertices.}
		
		\begin{enumerate}
			\item Run $\gamma$ independent instances of $\prealgosub(G,\ell,\emptyset)$ and let $(\Lsave)^{(1)},\dotsc,(\Lsave)^{(\gamma)}$ be the corresponding collections of lists.
			
			
			\item For each vertex $u \in V$, by running $\sampalgosub(V,u,\ell,\emptyset,(\Lsave)^{(j)})$ for each $j \in [\gamma]$, we take $\gamma$ independent samples from $\RW^{G}_{\ell}$. Let $w_{u}$ be the fraction of these random walks that revisit $u$ in $\ell$ steps.
			
			
			\item Let $\WTVheavy$ be the set of vertices with $w_u \ge 0.5$, and $\WTVlight$ be the set of vertices with $w_u < 0.5$.
		\end{enumerate}
		
		\State \textbf{Second Pass: heavy-light edge recording}
		
		\begin{enumerate}
			\item Let $\Vfull = \WTVheavy$.
			
			\item Run $\prealgosub(G, \gamma \cdot \ell,\Vfull)$ to obtain a collection of lists $\vecLsave$.
			
		\end{enumerate}
		
		\Ensure The set $\Vfull$ and the collection of lists $\vecLsave$.
	\end{algorithmic}
\end{algorithm}

\begin{algorithm}[H]
	\caption{Sampling phase of $\algotwopass$: $\sampalgomain(V,\ustart,L,\Vfull, \vecLsave = \{ \Lsave_v \}_{v \in V})$}
	\label{algo:two-pass-sampling}
	
	\begin{algorithmic}[1]
		\renewcommand{\algorithmicrequire}{\textbf{Input:}}
		\renewcommand{\algorithmicensure}{\textbf{Output:}}
		
		\Require A starting vertex $\ustart$. The path length $L \in \N$. A subset $\Vfull \subseteq V$, and a collection of lists $\vecLsave$.
		
		\Ensure Simulate $\sampalgosub(V,\ustart,L,\Vfull,\vecLsave)$ and return its output.
	\end{algorithmic}
	
\end{algorithm}


\paragraph*{Analysis of the algorithm.} We first show that with high probability, $\WTVlight$ and $\WTVheavy$ are subsets of $\Vlight$ and $\Vheavy$ respectively.

\begin{lemma}\label{lemma:subsample-light-and-heavy}
	Given a directed graph $G = (V,E)$ with $n$ vertices, $L \in \N$ and $\delta \in (0,1/n)$, letting $\ell = \sqrt{L}$, with probability at least $1 - \delta/2$ over the internal randomness of $\prealgomain(G,L,\delta)$, it holds that $\WTVlight \subseteq \Vlight$ and $\WTVheavy \subseteq \Vheavy$.
\end{lemma}
\begin{proof}
	Setting $c_1$ in~\autoref{algo:two-pass-preprocess} to be a large enough constant and applying~\autoref{cor:sub-procedure-correctness} and the Chernoff bound, with probability at least $1 - n \cdot \delta^3 \ge 1 - \delta/2$, $|w_u - \visit_{[1,\ell]}(u,u)| \le 0.1$ for every $u \in V$. The lemma then follows from the definition of heavy and light vertices.
\end{proof}


Next, we show that with high probability, a random walk does not visit a light vertex too many times.


\begin{lemma}\label{lemma:light-is-not-visited-much}
	Given a directed graph $G = (V,E)$ with $n$ vertices, $L \in \N$ and $\delta \in (0,1/n)$, letting $\ell = \sqrt{L}$ and $\gamma = c_1 \cdot \log \delta^{-1}$, where $c_1 > 1$ is the sufficiently large constant, for every vertex $\ustart \in V$ and vertex $v \in \Vlight^\ell(G)$, an $L$-step random walk starting from $\ustart$ visits $v$ more than $\gamma \cdot \ell$ times with probability at most $\delta / 2n$.
\end{lemma}
\begin{proof}
	Suppose we have an infinite random walk $\rv{W}$ starting from $\ustart$ in $G$. Letting $\tau = \gamma \ell$, the goal here is to bound the probability that during the first $L$ steps, $\rv{W}$ visits $v$ more than $\tau$ times. We denote this as the bad event $\calE_{\sf bad}$.
	
	Let $\rv{Z}_{i}$ be the random variable representing the step at which $\rv{W}$ visits $v$ for the $i$-th time (if $\rv{W}$ visits $v$ less than $i$ times in total, we let $\rv{Z}_i = \infty$). $\calE_{\sf bad}$ is equivalent to that $\rv{Z}_{\tau+1} \le L$.
	
	$\rv{Z}_{\tau+1} \le L$ further implies that for at least $(\tau - \ell)$ $i \in [\tau]$, $\rv{Z}_{i+1} - \rv{Z}_i \le \ell$ and $\rv{Z}_{i} < \infty$. In the following we denote this event as $\calE_1$ and bounds its probability instead. 
	
	For each $i \in [\tau]$, let $\rv{Y}_i$ be the random variable which takes value $1$ if both $Z_{i} < \infty$ and $\rv{Z}_{i+1} - \rv{Z}_i \le \ell$ hold, and $0$ otherwise. Letting $\rv{Y}_{<i} = (\rv{Y}_1,\dotsc,\rv{Y}_{i-1})$, the following claim is crucial for us.
	\begin{claim}
		For every $i \in[\tau]$ and every possible assignments $Y_{<i} \in \{0,1\}^{i-1}$, we have
		\[
		\E[\rv{Y}_i | \rv{Y}_{<i} = Y_{<i}] \le 2/3.
		\]
	\end{claim}
	\begin{proof}
		By the Markov property of the random walk, and noting that $\rv{Y}_i$ is always $0$ when $\rv{Z}_i = \infty$, we have.
		\begin{align*}
		\E[\rv{Y}_i | \rv{Y}_{<i} = Y_{<i}] &= \sum_{j=0}^{\infty} \Pr[\rv{Z}_i = j | \rv{Y}_{<i} = Y_{<i}] \cdot \E[\rv{Y}_i | \rv{Y}_{<i} = Y_{<i},\rv{Z}_i = j]\\
		&= \sum_{j=0}^{\infty} \Pr[\rv{Z}_i = j | \rv{Y}_{<i} = Y_{<i}] \cdot \E[\rv{Y}_i | \rv{Z}_i = j].
		\end{align*}
		
		To further bound the quantity above, recall that the event $\rv{Z}_i = j$ means that the random walk $\rv{W}$ starting from $\ustart$ visits the light vertex $v$ for the $i$-th time at $\rv{W}$'s $j$-th step, and we have
		\[
		\E[\rv{Y}_i | \rv{Z}_i = j] = \Pr[\rv{Y}_i = 1 | \rv{Z}_i = j] = \Pr[\rv{Z}_{i+1} \le j + \ell | \rv{Z}_i = j].
		\]		
		
		By the Markov property of the random walk $\rv{W}$, $\Pr[\rv{Z}_{i+1} \le j + \ell | \rv{Z}_i = j]$ equals the probability that a random walk starting from $v$ revisits $v$ in at most $\ell$ steps. By the definition of light vertices, we can bound that by $2/3$, which completes the proof.
		
	\end{proof}
	
	Then by the Azuma-Hoeffding inequality (\autoref{cor:azuma-hoeffding}),
	\begin{align*}
	\Pr_{\rv{W}}[\calE_{\sf bad}] 
	&\le \Pr_{\rv{W}}[ \calE_1 ] \\
	&= \Pr_{\rv{W}}\left[ \sum_{i=1}^{\tau} \rv{Y}_i \ge (\tau - \ell)\right]\\
	&\le \exp(-\Omega(\tau - \ell - 2/3 \cdot \tau)) \le \delta / 2n,
	\end{align*}
	the last inequality follows from the fact that $\gamma = c_1 \cdot \log \delta^{-1}$ for a sufficiently large constant $c_1$.
	
	
\end{proof}

\newcommand{\bVfull}{\rv{V}_{\sf full}}
The correctness of the algorithm is finally completed by the following theorem.
\begin{theorem}[Formal version of~\autoref{theo:algo-main}]
	Given a directed graph $G = (V,E)$ with $n$ vertices, $L \in \N$ and $\delta \in (0,1/n)$. Let $\bvecLsave$ and $\bVfull$ be the two random variables of the output of $\prealgomain(G,L,\delta)$. For every $\ustart \in V$, the following hold:
	
	\begin{itemize}
		\item The output distribution of $\sampalgomain(V,\ustart,L,\bVfull, \bvecLsave)$ has statistical distance at most $\delta$ from $\RW^G_L(\ustart)$. 
		
		\item Both of $\prealgomain(G,L,\delta)$ and $\sampalgomain(V,\ustart,L,\bVfull, \bvecLsave)$ use at most $O(n \cdot \sqrt{L} \cdot \log \delta^{-1})$ words of space.
	\end{itemize}
	
\end{theorem}

\begin{proof}
	Note that we can safely assume $L \le n^2$, since otherwise one can always use $O(n^2)$ words to store all the edges in the graph. In this case, we have that $L \le n \cdot \sqrt{L}$ and the space for restoring the $L$-step output walk can be ignored.
	
	Let $\bWTVheavy = \bVfull$ and $\bWTVlight = V \setminus \bWTVheavy$. Let $\calE_{\sf good}$ be the event that $\bWTVlight \subseteq \Vlight$ and $\bWTVheavy \subseteq \Vheavy$. By~\autoref{lemma:subsample-light-and-heavy}, we have that $\Pr[\calE_{\sf good}] \ge 1 - \delta / 2$.

	Now we condition on the event $\calE_{\sf good}$. In this case, it follows from~\autoref{lemma:heavy-total-out-degree-bounded} that $\prealgosub(G, \gamma \cdot \ell,\bWTVheavy)$ operates correctly (by setting the constant $c_2$ in~\autoref{algo:one-pass-preprocess} to be sufficiently large).
	
	By \autoref{lemma:light-is-not-visited-much} and a union bound, the probability of $\sampalgomain(V,\ustart,L,\bWTVheavy, \bvecLsave)$ outputs failure is at most $\delta / 2$. By~\autoref{lemma:sub-procedure-correctness}, it follows that the statistical distance between the output distribution of $\sampalgomain(V,\ustart,L,\bWTVheavy, \bvecLsave)$ and $\RW^G_L(\ustart)$ is at most $\delta / 2$.
	
	The theorem follows by combing the above with the fact that $\Pr[\calE_{\sf good}] \ge 1 - \delta / 2$.
	
\end{proof}

\subsection{Two-pass Streaming in the Turnstile Model}\label{sec:turnstile}

Similar to the algorithm in~\cite{Jin19}, our algorithms can also be easily adapted to work for the \emph{turnstile graph streaming model}, where both insertions and deletions of edges are allowed. Note that our two-pass algorithm $\algotwopass$ only accesses the input graph stream via the one-pass preprocessing subroutine $\prealgosub$. Hence, it suffices to implement $\prealgosub$ in the turnstile model as well. There are two distinct tasks in $\prealgosub$: (1) for light vertices, we need to sample their outgoing neighbors with replacement and (2) for heavy vertices, we need to record all their outgoing neighbors.

\paragraph*{Uniformly sampling via $\ell_1$ sampler.} For light vertices, uniformly sampling some out-neighbors from each vertex without replacement can be implemented via the following $\ell_1$ sampler in the turnstile model.

\begin{lemma}[$\ell_1$ sampler in the turnstile model~\cite{JayaramW18}]\label{lemma:ell-1-sampler}
	Let $n \in \N$, failure probability $\delta \in (0,1/2)$ and $f \in \R^n$ be a vector defined by a streaming of updates to its coordinates of the form $f_i \leftarrow f_i + \Delta$, where $\Delta \in \{-1,1\}$. There is a randomized algorithm which reads the stream, and with probability at most $\delta$ it outputs FAIL, otherwise it outputs an index $i \in [n]$ such that:
	\begin{align*}
	\Pr(i=j) = \frac{|f_j|}{\|f\|_1} + O(n^{-c}) , ~~~~~~ \forall j \in [n]
	\end{align*}
	where $c \geq 1$ is some arbitrarily large constant.

	The space complexity of this algorithm is bounded by $O(\log^2(n) \cdot \log(1/\delta) )$ bits in the random oracle model, and $O(\log^2(n) \cdot (\log \log n)^2 \cdot \log(1/\delta) )$ bits otherwise.
\end{lemma}

\begin{remark}
	To get error in the statistical distance also to be at most $\delta$, one can simply set $n$ to be larger than $1/\delta$. And in that case the space complexity can be bounded by $O(\log^4(n/\delta))$.
\end{remark}

\paragraph*{Recording all outgoing neighbors via $\ell_1$ heavy hitter.} For heavy vertices, recording all their outgoing neighbors can be implemented using the following $\ell_1$ heavy hitter in the turnstile model. (Recall that we assumed our graphs is a simple graph without multiple edges.)

\begin{lemma}[$\ell_1$ heavy hitter in the turnstile model~\cite{ccf02}]\label{lemma:ell-1-heavy-hitter}
	Let $n,k \in \N$, $\delta \in (0,0.1)$ and $f \in \R^n$ be a vector defined by a streaming of updates to its coordinates of the form $f_i \leftarrow f_i + \Delta$, where $\Delta \in \{-1,1\}$. There is an algorithm which reads the stream and returns a subset $L \subset [n]$ such that $i \in L$ for every $i \in [n]$ such that $|f_i| \ge \|f\|_1 / k$ and $i \not\in L$ for every $i \in [n]$ such that $|f_i| \le \|f\|_1 / 2k$. The failure probability is at most $\delta$, and the space complexity is at most $O(k \cdot \log (n) \cdot \log (n/\delta))$.
\end{lemma}

\paragraph*{Algorithm in the turnstile model.} Modifying $\prealgosub$ with ~\autoref{lemma:ell-1-sampler} and~\autoref{lemma:ell-1-heavy-hitter}, we can generalize our two-pass streaming algorithm to work in two-pass turnstile model.\footnote{In more details, for each light vertex $u$, we run $\tau$ independent copies of the $\ell_1$ sampler to obtain $\tau$ samples from its outgoing neighbors with replacement. We also let $k = c_2 \cdot \tau \cdot n$ and use the $\ell_1$ heavy hitter to record all outgoing neighbors for all heavy vertices in $\widetilde{O}(n \cdot \sqrt{L} \cdot \log(1/\delta))$ space.} 

\begin{remark}[Two-pass algorithm in the turnstile model]\label{remark:turnstile}
	There exists a streaming algorithm $\algo_{\sf turnstile}$ that given an $n$-vertex {\em directed} graph $G = (V,E)$ via a stream of both edge insertions and edge deletions, a starting vertex $\ustart \in V$, a non-negative integer $L$ indicating the number of steps to be taken, and an error parameter $\delta \in (0,1/n)$, satisfies the following conditions:
	\begin{enumerate}
		\item $\algo_{\sf turnstile}$ uses at most $\tilde{O}(n \cdot \sqrt{L} \cdot \log \delta^{-1})$ space and makes two passes over the input graph~$G$.
		\item $\algo_{\sf turnstile}$ samples from some distribution $\calD$ over $V^{L+1}$ satisfying $\tvd{\calD}{\RW^G_L(\ustart)} \le \delta$. 
	\end{enumerate}
\end{remark}

%% file: files/lowb.tex

\section{Proof of \autoref{theo:lowb-main-multi-pass}}
\label{sec:lb}

\begin{reminder}{\autoref{theo:lowb-main-multi-pass}} 
\thmstmtLB
\end{reminder}
\begin{proof}
We show \autoref{theo:lowb-main-multi-pass} in two steps, that are captured in \autoref{lemma:reduction} and \autoref{thm:lb} below. \autoref{theo:lowb-main-multi-pass} is a direct corollary of \autoref{lemma:reduction} and \autoref{thm:lb}.
\end{proof}

The following distribution is used in \autoref{lemma:reduction} and \autoref{thm:lb}. We sometimes omit the subscript $L$ when it is clear from context.

\begin{construction}{Hard Input Distribution $\dist_{n,p,L}$}
	
	\begin{itemize} 

		\item \textbf{Setup:} Define $\Delta = \paren*{ \frac{L}{ \paren*{ \log n }^{ 10^{10p} } } }^{\frac{1}{p}}$ and $\Delta' = \Delta / n$ .
		
		\item \textbf{Vertices:} We construct a layered graph $G$ with $p+2$ layers $V_1, V_2, \cdots, V_{p+2}$ satisfying $\card*{V_1} = 1$ and $\card*{V_i} = n$ for all $i \in [p+2] \setminus \{1\}$. We use $s$ to refer to the unique vertex in $V_1$.
		
		\item \textbf{Edges:} There are $p+2$ sets of edges $E_1, E_2, \cdots, E_{p+2}$ where edges $E_i$ are between $V_i$ and layer $V_{i+1}$ (indices taken modulo $p+2$). These are constructed as follows:
		\begin{itemize}
			\item The set $E_1$ is a singleton $\{(s, v)\}$ where $v$ is a vertex sampled uniformly at random from $V_2$.
			
			\item For all $1 < i < p + 2$, and all $v \in V_i$, sample a a subset $S_v \subseteq V_{i+1}$ uniformly and independently such that $\card*{ S_v } = \Delta$. Set $E_i = \set*{ (v, v') \mid v \in V_i, v' \in S_v }$.
			
			\item Define the set $E_{p+2} = \{(v, s) \mid v \in V_{p+2}\}$.			
		\end{itemize}
		
		\item \textbf{Edge ordering:} The edges are revealed to the streaming algorithm in the order $E_{p+2}, E_{p+1}, \cdots, E_1$.
		\end{itemize}
\end{construction}

\begin{lemma}
\label{lemma:reduction}
Suppose there exists a constant $\beta\in(0, 1]$, an integer $p\geq 1$, integers $n, L$ that are sufficiently large and satisfy $L=\lceil n^{\beta}\rceil$ such that there exists a (randomized) $p$-pass streaming algorithm $\A$ that takes space $\frac{n \cdot L^{1/p}}{ \paren*{ \log n }^{ 10^{20p} } }$ and, given an $n$-vertex {\em directed} graph $G = (V,E)$ and a starting vertex $\ustart \in V$, can sample from a distribution $\calD$ such that $\tvd{\calD}{\RW^G_L(\ustart)} \le 1 - \frac{1}{ \paren*{ \log n }^{10} }$.

Then, there exists another randomized $p$-pass streaming algorithm $\A'$ that takes space $\frac{n \cdot L^{1/p}}{ \paren*{ \log n }^{ 10^{15p} } }$ and satisfies:
\[
\Pr_{\rv{G} \sim \dist_{n,p}}\paren*{ \A'\paren*{ \rv{G} } = \rv{P}^{p+1}(s) } \geq \frac{1}{ \paren*{ \log n }^{20} } .
\]
\end{lemma}
\begin{proof}
Let $\A'$ be the algorithm that first runs $\A$ on its input and $s$ to get as output a walk $W = (v_0,v_1,\dotsc,v_L)$. Define $E(W) = \set*{ (v_{i-1}, v_i) \mid i \in [L] }$ to be the set of edges witnessed by $W$. The algorithm $\A'$ then outputs all paths $P^{p+1}(s, W)$ of length $p+1$ starting for $s$ using only the edges $E(W)$. 

Let $N^{\leq p}(s) = \bigcup_{i' = 0}^p N^i(s)$. Observe that if a walk $W$ satisfies $P^{p+1}(s, W) \neq P^{p+1}(s)$, then either $E(W) \not\subseteq E$ or $P(N^{\leq p}(s)) \not\subseteq E(W)$. Thus, we have, for all $G \in \supp(\dist_{n,p})$ that:
\begin{align*}
\Pr\paren*{ \A'\paren*{ G } \neq P^{p+1}(s) } &= \Pr_{W \sim \A(G)}\paren*{ P^{p+1}(s, W) \neq P^{p+1}(s) } \\
&\leq \Pr_{W \sim \RW^G_L(s)}\paren*{ P^{p+1}(s, W) \neq P^{p+1}(s) } + 1 - \frac{1}{ \paren*{ \log n }^{10} } \\
&\leq \Pr_{W \sim \RW^G_L(s)}\paren*{ E(W) \not\subseteq E \vee P(N^{\leq p}(s)) \not\subseteq E(W) } + 1 - \frac{1}{ \paren*{ \log n }^{10} } \tag{Union bound} .
\end{align*}
As $E(W) \subseteq E$ for all $W \sim \RW^G_L(s)$, we have:
\[
\Pr\paren*{ \A'\paren*{ G } \neq P^{p+1}(s) } \leq \Pr_{W \sim \RW^G_L(s)}\paren*{ P(N^{\leq p}(s)) \not\subseteq E(W) } + 1 - \frac{1}{ \paren*{ \log n }^{10} } .
\]
Thus, to finish the proof, it suffices to show that $\Pr_{W \sim \RW^G_L(s)}\paren*{ P(N^{\leq p}(s)) \not\subseteq E(W) } \leq \frac{1}{ \paren*{ \log n }^{20} }$. This is done in the rest of the proof. First, observe from the definition of $\dist_{n,p}$ that $P(N^{\leq p}(s))$ is a collection of at most $p \cdot \Delta^p \ll \frac{L}{ \paren*{ \log n }^{ 10^{8p} } }$ edges. We get by a union bound:
\begin{equation}
\label{eq:reduction1}
\Pr_{W \sim \RW^G_L(s)}\paren*{ P(N^{\leq p}(s)) \not\subseteq E(W) } \leq L \cdot \max_{e \in P(N^{\leq p}(s))} \Pr_{W \sim \RW^G_L(s)}\paren*{ e \notin E(W) } .
\end{equation}
Fix $e \in P(N^{\leq p}(s))$ and observe that $v_i = s$ for every $i$ that is a multiple of $p+2$. Using the Markov property of random walks, we get:
\[
\Pr_{W \sim \RW^G_L(s)}\paren*{ e \notin E(W) } \leq \paren*{ \Pr_{W \sim \RW^G_L(s)}\paren*{ \forall i \in [p+2] : e \neq (v_{i-1}, v_i) } }^{ \frac{L}{10p} } .
\]
As the out-degree of $s$ is $1$ and the out-degree of every other vertex is at most $\Delta$ (\autoref{obs:smallnbr}), we conclude that: 
\[
\Pr_{W \sim \RW^G_L(s)}\paren*{ e \notin E(W) } \leq \paren*{ 1 - \frac{1}{ \Delta^p } }^{ \frac{L}{10p} } \leq \e^{ -\frac{L}{10p \cdot \Delta^p} } \leq \frac{1}{n^{20}} ,
\]
as $p \cdot \Delta^p \ll \frac{L}{ \paren*{ \log n }^{ 10^{8p} } }$. Plugging into \autoref{eq:reduction1} finishes the proof.

\end{proof}

\newcommand\statementthmlb{ 
Let a constant $\beta \in (0,1]$ and an integer $p \ge 1$ be given. Let $n \geq 1$ be sufficiently large and $L=\lceil n^{\beta}\rceil$. For all (randomized) $p$-pass streaming algorithms $\A$ that takes space $\frac{n \cdot L^{1/p}}{ \paren*{ \log n }^{ 10^{15p} } }$, we have that:
\[
\Pr_{\rv{G} \sim \dist_{n,p}}\paren*{ \A\paren*{ \rv{G} } = \rv{P}^{p+1}(s) } \leq \frac{1}{ \paren*{ \log n }^{50} } ,
\]
}
\begin{theorem}
\label{thm:lb}
\statementthmlb
\end{theorem}

The proof of \autoref{thm:lb} spans the rest of this section.  We start with some notation and some properties of the distribution $\dist_{n,p}$. Fix $p \geq 1$ and $n$ large enough (as a function of $p$) for the rest of this subsection.

\subsection{Properties of $\dist_{n,p}$} \label{sec:distprop}

\paragraph*{Notation.} As the set $E_{p+2}$ is fixed, we shall sometimes view $\dist_{n,p}$ as a distribution over the sets $E_1, E_2, \cdots, E_{p+1}$. We shall use $V$ to denote the set of all vertices and $E$ to denote the set of all edges, thus $G = (V, E)$. For $i \in [p + 1]$ we define $E_{-i}$ to be $E \setminus E_i$.

For a vertex $v \in V$ and $i \geq 0$, define the set $P^i(v)$ to be the set of paths of length $i$ starting from $v$. Also define $P^i(S)$ for a subset $S \subseteq V$ of vertices as $P^i(S) = \bigcup_{v \in S} P^i(v)$. We drop the superscript $i$ when $i = 1$. Observe that for all $i \in [p+2]$, we have $P(V_i) = E_i$. Similarly, define $N^i(v)$ to be the set of all {\em vertices} that can be reached by a path of length exactly $i$ from $v$, {\em i.e.}, a vertex $v' \in N^i(v)$ if and only if there is a path ending at $v'$ in $P^i(v)$. 

Throughout, we shall use $\h(x) = - x \log(x) - (1-x) \log (1-x)$ to denote the binary entropy function. Observe that $\h(\cdot)$ is concave and monotone increasing for $0 < x < \frac{1}{2}$.

In this section, we collect some useful properties of the distribution $\dist_{n,p}$ defined above. All these properties can be proved by straightforward but tedious calculations, so we defer their proofs to~\autoref{sec:missing-proofs-distprop}.

\subsubsection{Size of $\rv{N}^k(s)$}
\label{sec:distprop1}

We state without proof the following observation:
\begin{observation}
\label{obs:smallnbr}
It holds that:
\begin{enumerate}
\item \label{item:obs:smallnbr1} For all $v \in V_1 \cup V_{p+2}$, we have $\card*{ \rv{N}(v) } = 1$. 
\item \label{item:obs:smallnbr2} For all $v \in V \setminus \paren*{ V_1 \cup V_{p+2} }$, we have $\card*{ \rv{N}(v) } = \Delta$. 
\item \label{item:obs:smallnbr3} For all $k \in [p+1]$, we have $\card*{ \rv{N}^k(s) } \leq \Delta^{k-1}$. 
\end{enumerate}
\end{observation}

Owing to \autoref{item:obs:smallnbr3} above, it shall be useful to define, for all $k \in [p+1]$, the notation
\begin{equation}
\label{eq:deltak}
\Delta_k = \Delta^{k-1} \hspace{1cm} \text{and} \hspace{1cm} \Delta'_k = \frac{ \Delta_k }{n} .
\end{equation}
Also recall that we used $\Delta'$ to denote $\Delta / n$.

\begin{lemma}
\label{lemma:largenbr}
For all $k \in [p+1]$, we have:
\[
\Pr\paren*{ \card*{ \rv{N}^k(s) } \leq \Delta_k \cdot \paren*{ 1 - \frac{2k}{ \paren*{ \log n }^{ 10^{ 10p } } } } } \leq \frac{k}{n^{200}} .
\]
\end{lemma}

\begin{corollary}[Corollary of \autoref{obs:smallnbr} and \autoref{lemma:largenbr}]
\label{cor:nbr}
For all $k \in [p+1]$, we have:
\[
\Pr\paren*{ \Delta_k \cdot \paren*{ 1 - \frac{1}{ \paren*{ \log n }^{ 10^{ 10p } - 2 } } } \leq \card*{ \rv{N}^k(s) } \leq \Delta_k } \geq 1 - \frac{1}{n^{150}} .
\]
\end{corollary}

\subsubsection{Entropy of $\rv{N}^k(s)$}
\label{sec:distprop2}

\begin{lemma}
\label{lemma:entropynbrub}
For all $v \in V \setminus \paren*{ V_1 \cup V_{p+2} }$ and any event $E$, we have:
\begin{align*}
\H\paren*{ \rv{N}(v) \mid E } &\leq n \cdot \h\paren*{ \Delta' } \cdot \paren*{ 1 + \frac{1}{ \paren*{ \log n }^{ 10^{10p} } } } \text{ and} \\
\H\paren*{ \rv{N}(v) } &\geq n \cdot \h\paren*{ \Delta' } \cdot \paren*{ 1 - \frac{1}{ \paren*{ \log n }^{ 10^{10p} } } } .
\end{align*}
\end{lemma}
\begin{corollary}
\label{cor:entropynbrub}
For all $1 < k \leq p+1$ and any event $E$, we have:
\begin{align*}
\H\paren*{ \rv{E}_k \mid E } &\leq n^2 \cdot \h\paren*{ \Delta' } \cdot \paren*{ 1 + \frac{1}{ \paren*{ \log n }^{ 10^{10p} } } } \text{ and}\\
\H\paren*{ \rv{E}_k } &\geq n^2 \cdot \h\paren*{ \Delta' } \cdot \paren*{ 1 - \frac{1}{ \paren*{ \log n }^{ 10^{10p} } } } .
\end{align*}
\end{corollary}
 
\begin{lemma}
\label{lemma:entropynbrlb}
We have $\H\paren*{ \rv{N}(s) } = \log n$ and, for all $1 < k \leq p+1$:
\[
\H\paren*{ \rv{N}^k(s) } \geq n \cdot \h\paren*{ \Delta'_k } \cdot \paren*{ 1 - \frac{1}{ \paren*{ \log n }^{ 10^{8p} } } } .
\]
\end{lemma}

\subsubsection{Entropy of $\rv{P}^k(s)$}
\label{sec:distprop3}

\begin{lemma}
\label{lemma:entropypath}
For all events $E$, we have $\H\paren*{ \rv{P}(s) \mid E } \leq \log n$ and, for all $1 < k \leq p+1$:
\[
\H\paren*{ \rv{P}^k(s) \mid E } \leq \Delta'_{k-1} \cdot \paren*{ 1 + \frac{1}{ \paren*{ \log n }^{ 10^{9p} } } } \cdot \H\paren*{ \rv{E}_k }.
\]
\end{lemma}
\begin{lemma}
\label{lemma:entropypathlb}
We have $\H\paren*{ \rv{P}(s) } = \log n$ and, for all $1 < k \leq p+1$:
\[
\H\paren*{ \rv{P}^k(s) } \geq \Delta'_{k-1} \cdot \paren*{ 1 - \frac{1}{ \paren*{ \log n }^{ 10^{9p} } } } \cdot \H\paren*{ \rv{E}_k }.
\]
\end{lemma}

\begin{lemma}
\label{lemma:minentropypath}
For all events $E$, it holds that:
\[
2^{ - \H_{\infty}\paren*{ \rv{P}^{p+1}(s) \mid E } } \leq 1 - \frac{ \H\paren*{ \rv{P}^{p+1}(s) \mid E } - 1 }{ \Delta'_p \cdot \paren*{ 1 + \frac{1}{ \paren*{ \log n }^{ 10^{9p} } } } \cdot \H\paren*{ \rv{E}_{p+1} } } .
\]
\end{lemma}

\subsection{The Communication Lower Bound} 
\label{sec:cclb}

\begin{reminder}{\autoref{thm:lb}}
\statementthmlb
\end{reminder}

\paragraph*{Communication game.} To show our lower bound, we consider a communication game $\Pi$ with $p+1$ players. The edges $E_{p+2}$ are known to all players. In addition, player $i \in [p+1]$ also knows the edges $E_{p+2-i}$. Define $T = p(p+1)$. The communication takes place in $T - 1$ rounds, where in round $i$, player $i$ sends a message $M_i$ to player $i+1$ (indices taken modulo $p+1$) based on its input and all the messages received so far. After these $T - 1$ rounds have taken place, player $p+1$ outputs an answer based on its input and all the received messages. We treat the output as the $T^{\text{th}}$ message and denote it by $M_T$. We use $\norm*{\Pi}$ to denote the maximum (over all inputs) total communication in $\Pi$ (excluding the output).


\begin{proof}[Proof of \autoref{thm:lb}]
To start, note that we can assume that $\A$ is deterministic without loss of generality. The algorithm $\A$ implies a deterministic communication game $\Pi$ as above satisfying 
\begin{align*}
\norm*{\Pi} \leq T \cdot \frac{n \cdot L^{1/p}}{ \paren*{ \log n }^{ 10^{15p} } } \leq \frac{n \cdot \Delta}{ \paren*{ \log n }^{ 10^{10p} } }.
\end{align*}
For $j \in [T]$, We shall use $\rv{M}_j$ to denote the random variable corresponding to the $j^{\text{th}}$ message of the protocol. This random variable is over the probability space defined by $\dist_{n,p}$. Also, define $\rv{M}_{<j}$ and $\rv{M}_{\leq j}$ to be the random variables $\rv{M}_1 \cdots \rv{M}_{j-1}$ and $\rv{M}_1 \cdots \rv{M}_j$ respectively. 
For $k \in [p+1]$, define:
\[
\eps_k = \paren*{ \frac{1}{ \log n } }^{ 10^{ 5(p+1) - 2k } } .
\]

Observe that $\norm*{\Pi} \leq \eps_1^2 \cdot \H\paren*{ \rv{E}_k }$. We first show the following conditional independence result:

\begin{lemma}
\label{lemma:ind}
For all $0 \leq j \leq T$ and $i \in [p + 1]$, we have:
\[
\I\paren*{ \rv{E}_i : \rv{E}_{-i} \mid \rv{M}_{\leq j} } = 0 .
\]
\end{lemma}
\begin{proof}
We repeatedly apply \autoref{lemma:removecondition} to remove the conditioning on $\rv{M}_{\leq j}$. This is possible as for all $j' \in [j]$, either message $j'$ is not sent by player $p + 2 - i$ in which case $\rv{M}_{j'}$ is independent of $\rv{E}_i$ given $\rv{E}_{-i}$ and $\rv{M}_{<j'}$, or message $j'$ is sent by player $p + 2 - i$, in which case $\rv{M}_{j'}$ is independent of $\rv{E}_{-i}$ given $\rv{E}_i$ and $\rv{M}_{<j'}$. Using \autoref{lemma:removecondition}, we have:
\[
\I\paren*{ \rv{E}_i : \rv{E}_{-i} \mid \rv{M}_{\leq j} } \leq \I\paren*{ \rv{E}_i : \rv{E}_{-i} } = 0 ,
\]
by definition of $\dist_{n,p}$.

\end{proof}

The following lemma shows that the entropy of the set $\rv{P}^{p+1}(s)$ remains high after knowing all messages $M_{\le (p+1) p}$.

\begin{lemma}
\label{lemma:induction}
For all $1 \le k \leq p + 1$, we have:
\begin{align*}
\H\paren*{ \rv{P}^k(s) \mid \rv{M}_{\leq kp} } &= \log n , &\text{~if~} k = 1\\
\H\paren*{ \rv{P}^k(s) \mid \rv{M}_{\leq kp} } &\geq \paren*{ 1 - \eps_k } \cdot \Delta'_{k-1} \cdot \H\paren*{ \rv{E}_k } , &\text{~if~} k > 1.
\end{align*}
\end{lemma}

Before proving~\autoref{lemma:induction}, we need the following two important technical lemmas, whose proofs can be found in~\autoref{sec:missing-proofs-cclb}. Let $v^{(1)}_k, v^{(2)}_k, \cdots, v^{(n)}_k$ be the vertices in layer $V_k$ and, for $i \in [n]$, define the set $V^{<i}_k = \set*{ v^{(1)}_k, \cdots, v^{(i-1)}_k }$, and define $t = kp$ and $t' = (k-1)p$.

\begin{lemma}
	\label{lemma:bad1}
	For all $1 < k \leq p + 1$, assuming~\autoref{lemma:induction} holds for $k-1$, there exists a set $\B_1 \subseteq \supp\paren*{ \rv{M}_{\leq (k-1)p} }$ such that $\Pr\paren*{ \rv{M}_{\leq (k-1)p} \in \B_1 } \leq \sqrt{ \eps_{k-1} }$ and for all $M_{\leq (k-1)p} \notin \B_1$, we have:
	\begin{align*}
	&\H\paren*{ \rv{N}^{k-1}(s) \mid M_{\leq (k-1)p} } \geq \H\paren*{ \rv{N}^{k-1}(s) } - 10 \cdot \sqrt{ \eps_{k-1} } \cdot \Delta'_{k-2} \cdot \H\paren*{ \rv{E}_{k-1} } , &\text{~if~} k > 2 \\
	&\H\paren*{ \rv{N}^{k-1}(s) \mid M_{\leq (k-1)p} } = \H\paren*{ \rv{N}^{k-1}(s) } = \log n, &\text{~if~} k = 2.
	\end{align*}
\end{lemma}

\begin{lemma}
	\label{lemma:size}
	For all $1 < k \leq p + 1$, assuming~\autoref{lemma:induction} holds for $k-1$, and let $\B_1$ be the set promised by \autoref{lemma:bad1}. For all $M_{\leq t'} \notin \B_1$, we have 
	\[
	\card*{ \set*{ i \in [n] \mid \Pr\paren*{ v^{(i)}_k \in \rv{N}^{k-1}(s) \mid M_{\leq t'} } \leq \frac{ \Delta'_{k-1} }{ 1 + \eps_k^5 } } } \leq \eps_k^5 \cdot n.
	\]
\end{lemma}

Now we are ready to prove~\autoref{lemma:induction}.

\begin{proofof}{\autoref{lemma:induction}}
Induction on $k$. For the base case, we have $k = 1$. Recall that $\rv{M}_{\leq p}$ is determined by $\rv{E}_{-1}$ while $\rv{P}(s)$ is determined by $\rv{E}_1$. As these two are independent, we have that:
\[
\H\paren*{ \rv{P}(s) \mid \rv{M}_{\leq p} } = \H\paren*{ \rv{P}(s) } = \log n ,
\]
as desired. For the induction step, we show the lemma holds for $k > 1$ assuming it holds for $k-1$. 

Let $\B_1$ be the set promised by \autoref{lemma:bad1}. Define, for all $M_{\leq t'}$, the set:
\begin{equation}
\label{eq:set}
S( M_{\leq t'} ) = \set*{ i \in [n] \mid \Pr\paren*{ v^{(i)}_k \in \rv{N}^{k-1}(s) \mid M_{\leq t'} } \leq \frac{ \Delta'_{k-1} }{ 1 + \eps_k^5 } } .
\end{equation}

By~\autoref{lemma:size}, for all $M_{\leq t'} \notin \B_1$ we have $|S( M_{\leq t'} )| \le \eps_k^5 \cdot n$. Letting 
\[
\nabla \coloneqq \H\paren*{ \rv{P}^k(s) \mid \rv{M}_{\leq t} }
\]
for simplicity, we have:
\begin{align*} 
\nabla&\geq \H\paren*{ \rv{P}(\rv{N}^{k-1}(s)) \mid \rv{M}_{\leq t} \rv{E}_{-k} } \tag{\autoref{lemma:cond-reduce} and $\rv{P}^k(s)$ determines $\rv{P}(\rv{N}^{k-1}(s))$} \\
&\geq \H\paren*{ \rv{P}(\rv{N}^{k-1}(s)) \mid \rv{M}_{\leq t'} \rv{E}_{-k} } \tag{As $\paren*{ \rv{M}_{\leq t'}, \rv{E}_{-k} }$ and $\paren*{ \rv{M}_{\leq t}, \rv{E}_{-k} }$ determine each other} \\
&\geq \sum_{ M_{\leq t'} } \Pr\paren*{ M_{\leq t'} } \cdot \sum_{ E_{-k} } \Pr\paren*{ E_{-k} \mid M_{\leq t'} } \cdot \H\paren*{ \rv{P}(N^{k-1}(s)) \mid M_{\leq t'}, E_{-k} } \tag{\autoref{def:cond-entropy} and $E_{-k}$ determines $N^{k-1}(s)$} .
\end{align*}

To continue, note that $\rv{P}(N^{k-1}(s))$ is determined by $\rv{E}_k$ and is independent of $\rv{E}_{-k}$ conditioned on $\rv{M}_{\leq t'}$ by \autoref{lemma:ind}. We get:
\begin{align*}
\nabla&\geq \sum_{ M_{\leq t'} } \Pr\paren*{ M_{\leq t'} } \cdot \sum_{ E_{-k} } \Pr\paren*{ E_{-k} \mid M_{\leq t'} } \cdot \H\paren*{ \rv{P}(N^{k-1}(s)) \mid M_{\leq t'} } \\
	  &\geq \sum_{ M_{\leq t'} } \Pr\paren*{ M_{\leq t'} } \cdot \sum_{ E_{-k} } \Pr\paren*{ E_{-k} \mid M_{\leq t'} } \cdot \sum_{ \substack{ i \in [n] \\ v^{(i)}_k \in N^{k-1}(s) } } \H\paren*{ \rv{P}(v^{(i)}_k) \mid \rv{P}(V^{<i}_k), M_{\leq t'} } \tag{\autoref{lemma:entropy-chain}, \autoref{lemma:cond-reduce}} \\
	  &\geq \sum_{ M_{\leq t'} } \Pr\paren*{ M_{\leq t'} } \cdot \sum_{i = 1}^n \Pr\paren*{ v^{(i)}_k \in \rv{N}^{k-1}(s) \mid M_{\leq t'} } \cdot \H\paren*{ \rv{P}(v^{(i)}_k) \mid \rv{P}(V^{<i}_k), M_{\leq t'} }.
\end{align*}

\newcommand{\rvZ}{\rv{Z}}

To ease the notation, for each $i \in [n]$, we define $\rvZ_i = \rv{P}(v_k^{(i)})$, and $\rvZ_{<i} = \{\rvZ_1,\dotsc,\rvZ_{i-1}\}$. By only considering $M_{\leq t'} \notin \B_1$ and $i \notin S(M_{\leq t'})$, we have:

\begin{align*}
\nabla&\geq \sum_{ M_{\leq t'} \notin \B_1 } \Pr\paren*{ M_{\leq t'} } \cdot \sum_{ i \notin S( M_{\leq t'} ) } \Pr\paren*{ v^{(i)}_k \in \rv{N}^{k-1}(s) \mid M_{\leq t'} } \cdot \H\paren*{ \rvZ_i \mid \rvZ_{<i}, M_{\leq t'} } \\
	  &\geq \frac{ \Delta'_{k-1} }{ 1 + \eps_k^5 } \cdot \sum_{ M_{\leq t'} \notin \B_1 } \Pr\paren*{ M_{\leq t'} } \cdot \sum_{ i \notin S( M_{\leq t'} ) } \H\paren*{ \rvZ_i \mid \rvZ_{<i}, M_{\leq t'} } \tag{\autoref{eq:set}}.
\end{align*}

Now, note that by~\autoref{lemma:entropy-chain}, we have:
\begin{align*}
\H\paren*{ \rv{E}_k \mid M_{\leq t'} } &= \sum_{i=1}^{n} \H\paren*{ \rvZ_i \mid \rvZ_{<i}, M_{\leq t'} } \\
									   &= \sum_{i \in S( M_{\leq t'} )} \H\paren*{ \rvZ_i \mid \rvZ_{<i}, M_{\leq t'} }	+
									      \sum_{i \not\in S( M_{\leq t'} )} \H\paren*{ \rvZ_i \mid \rvZ_{<i}, M_{\leq t'}}.
\end{align*}
Plugging in, we have:
\begin{align*}
\nabla\geq \frac{ \Delta'_{k-1} }{ 1 + \eps_k^5 } \cdot \sum_{ M_{\leq t'} \notin \B_1 } \Pr\paren*{ M_{\leq t'} } \cdot \paren*{ \H\paren*{ \rv{E}_k \mid M_{\leq t'} } - \sum_{ i \in S( M_{\leq t'} ) } \H\paren*{ \rvZ_i \mid \rvZ_{<i}, M_{\leq t'} } }\tag{\autoref{lemma:entropy-chain}} .
\end{align*}

Next, observe from \autoref{lemma:cond-reduce} that $\H\paren*{ \rvZ_i \mid \rvZ_{<i}, M_{\leq t'} } \leq \H\paren*{ \rvZ_i \mid M_{\leq t'} }$. From \autoref{lemma:entropynbrub}, we can extend as $\H\paren*{ \rvZ_i \mid \rvZ_{<i}, M_{\leq t'} } \leq 2n \cdot \h\paren*{ \Delta' }$. Finally, using \autoref{lemma:size}, we get:
\begin{align*} 
\nabla &\geq \frac{ \Delta'_{k-1} }{ 1 + \eps_k^5 } \cdot \sum_{ M_{\leq t'} \notin \B_1 } \Pr\paren*{ M_{\leq t'} } \cdot \paren*{ \H\paren*{ \rv{E}_k \mid M_{\leq t'} } - \eps_k^5 \cdot 2n^2 \cdot \h\paren*{ \Delta' } } \\
	   &\geq \frac{ \Delta'_{k-1} }{ 1 + \eps_k^5 } \cdot \paren*{ \H\paren*{ \rv{E}_k \mid \rv{M}_{\leq t'} } - \eps_k^5 \cdot 2n^2 \cdot \h\paren*{ \Delta' } - \sum_{ M_{\leq t'} \in \B_1 } \Pr\paren*{ M_{\leq t'} } \cdot \H\paren*{ \rv{E}_k \mid M_{\leq t'} } } \tag{\autoref{def:cond-entropy}}. 
\end{align*}

Recall that by~\autoref{cor:entropynbrub}, we have, for all $M_{\leq t'}$, that:
\begin{align}
\H\paren*{ \rv{E}_k \mid M_{\leq t'} } &\le 2 \cdot n^2 \cdot h(\Delta') \label{eq:rvE-cond-upper}\\
\H\paren*{ \rv{E}_k} &\ge n^2 \cdot h(\Delta') \cdot (1 - (\eps_k)^5) \label{eq:rvE-lower}.
\end{align}

Plugging in,
\begin{align*}
\nabla &\geq \frac{ \Delta'_{k-1} }{ 1 + \eps_k^5 } \cdot \paren*{ \H\paren*{ \rv{E}_k \mid \rv{M}_{\leq t'} } - \eps_k^5 \cdot 2n^2 \cdot \h\paren*{ \Delta' } - 2 \cdot \Pr\paren*{ \rv{M}_{\leq t'} \in \B_1 } \cdot n^2 \cdot \h\paren*{ \Delta' } } \tag{\autoref{eq:rvE-cond-upper}} \\
&\geq \frac{ \Delta'_{k-1} }{ 1 + \eps_k^5 } \cdot \paren*{ \H\paren*{ \rv{E}_k \mid \rv{M}_{\leq t'} } - \eps_k^5 \cdot 2n^2 \cdot \h\paren*{ \Delta' } - \eps_k^{5} \cdot n^2 \cdot \h\paren*{ \Delta' } } \tag{As $\Pr\paren*{ \rv{M}_{\leq t'} \in \B_1 } \leq \sqrt{ \eps_{k-1} } \le \eps_k^{5}/2$} \\
&\geq \frac{ \Delta'_{k-1} }{ 1 + \eps_k^5 } \cdot \paren*{ \H\paren*{ \rv{E}_k \mid \rv{M}_{\leq t'} } - \eps_k^3 \cdot \H\paren*{ \rv{E}_k } } \tag{\autoref{eq:rvE-lower} and $\eps_k^3 \ge 10 \cdot \eps_k^5$} \\
&\geq \frac{ \Delta'_{k-1} }{ 1 + \eps_k^5 } \cdot \paren*{ \H\paren*{ \rv{E}_k } \cdot \paren*{ 1 - \eps_k^3 } - \H\paren*{ \rv{M}_{\leq t'} } } \tag{\autoref{lemma:entropy-chain}} \\
&\geq \paren*{ 1 - \eps_k } \cdot \Delta'_{k-1} \cdot \H\paren*{ \rv{E}_k } \tag{As $\H\paren*{ \rv{M}_{\leq t'} } \leq \norm*{\Pi} \leq \eps_1^2 \cdot \H\paren*{ \rv{E}_k }$} .
\end{align*}

\end{proofof}

With the help of \autoref{lemma:induction}, we now continue the proof of \autoref{thm:lb}. As \autoref{lemma:induction} holds for $k = p+1$, we have that $\H\paren*{ \rv{P}^{p+1}(s) \mid \rv{M}_{\leq T} } \geq \paren*{ 1 - \eps_{p+1} } \cdot \Delta'_p \cdot \H\paren*{ \rv{E}_{p+1} }$. The following lemma is analogous to \autoref{lemma:bad1}.
\begin{lemma}
	\label{lemma:bad2}
	There exists a set $\B^* \subseteq \supp\paren*{ \rv{M}_{\leq T} }$ such that $\Pr\paren*{ \rv{M}_{\leq T} \in \B^* } \leq \sqrt{ \eps_{p+1} }$ and for all $M_{\leq T} \notin \B^*$, we have:
	\[
	\H\paren*{ \rv{P}^{p+1}(s) \mid M_{\leq T} } \geq \H\paren*{ \rv{P}^{p+1}(s) } - 10 \cdot \sqrt{ \eps_{p+1} } \cdot \Delta'_p \cdot \H\paren*{ \rv{E}_{p+1} } .
	\]
\end{lemma}

Let $\B^*$ be the set from \autoref{lemma:bad2}. We apply \autoref{lemma:minentropypath} for all $M_{\leq T} \notin \B^*$ to get:
\begin{align*}
2^{ - \H_{\infty}\paren*{ \rv{P}^{p+1}(s) \mid M_{\leq T} } } &\leq 1 - \frac{ \H\paren*{ \rv{P}^{p+1}(s) } - 15 \cdot \sqrt{ \eps_{p+1} } \cdot \Delta'_p \cdot \H\paren*{ \rv{E}_{p+1} } }{ \Delta'_p \cdot \paren*{ 1 + \epsilon_{p+1} } \cdot \H\paren*{ \rv{E}_{p+1} } } \\
&\leq 1 - \frac{ 1 - \epsilon_{p+1} - 15 \cdot \sqrt{ \eps_{p+1} } }{ 1 + \epsilon_{p+1} } \tag{\autoref{lemma:entropypathlb}} \\
&\leq 20 \cdot \sqrt{ \eps_{p+1} } .
\end{align*}

It follows that:
\[
\Pr_{\rv{G} \sim \dist_{n,p}}\paren*{ \A\paren*{ \rv{G} } = \rv{P}^{p+1}(s) } \leq 25 \cdot \sqrt{ \eps_{p+1} } \leq \frac{1}{ \paren*{ \log n }^{50} } .
\]

\end{proof}

  

%% file: files/info.tex

\section{Preliminaries in Information Theory}
\label{sec:info}

Throughout this subsection, we use sans-serif letters to denote random variables and reserve $E$ to denote an arbitrary event. All random variables will be assumed to be discrete and we shall adopt the convention $0 \log \frac{1}{0} = 0$. All logarithms are taken with base $2$. 

\subsection{Entropy}

\begin{definition}[Entropy] 
\label{def:entropy}
The (binary) entropy of $\rv{X}$ is defined as:
\[
\H(\rv{X}) = \sum_{x \in \supp(\rv{X})} \Pr(x) \cdot \log \frac{1}{\Pr(x)} .
\]
The entropy of $\rv{X}$ conditioned on $E$ is defined as:
\[
\H(\rv{X} \mid E) = \sum_{x \in \supp(\rv{X})} \Pr(x \mid E) \cdot \log \frac{1}{\Pr(x \mid E)} .
\]

\end{definition}

\begin{definition}[Conditional Entropy] 
\label{def:cond-entropy}
We define the conditional entropy of $\rv{X}$ given $\rv{Y}$ and $E$ as:
\[
\H(\rv{X} \mid \rv{Y}, E) = \sum_{y \in \supp(\rv{Y})} \Pr(y \mid E) \cdot \H(\rv{X} \mid \rv{Y} = y, E) .
\]
\end{definition}

Henceforth, we shall omit writing the $\supp(\cdot)$ when it is clear from context.

\begin{lemma}[Chain Rule for Entropy]
\label{lemma:entropy-chain}
It holds for all $\rv{X}$, $\rv{Y}$ and $E$ that:
\[
\H(\rv{X} \rv{Y} \mid E) = \H(\rv{X} \mid E) + \H(\rv{Y} \mid \rv{X}, E) .
\]
\end{lemma}
\begin{proof}
We have:
\begin{align*}
\H(\rv{X} \rv{Y} \mid E) &= \sum_{x,y} \Pr(x,y \mid E) \cdot \log \frac{1}{\Pr(x,y \mid E)} \\
&= \sum_{x,y} \Pr(x,y \mid E) \cdot \log \frac{1}{\Pr(x \mid E) \cdot \Pr(y \mid x, E)} \\
&= \sum_{x,y} \Pr(x,y \mid E) \cdot \log \frac{1}{\Pr(x \mid E)} + \sum_{x,y} \Pr(x,y \mid E) \cdot \log \frac{1}{\Pr(y \mid x, E)} \\
&= \H(\rv{X} \mid E) + \sum_x \Pr(x \mid E) \cdot \sum_y \Pr(y \mid x,E) \cdot \log \frac{1}{\Pr(y \mid x, E)} \\
&= \H(\rv{X} \mid E) + \H(\rv{Y} \mid \rv{X}, E) . \qedhere
\end{align*}
\end{proof}

\begin{lemma}[Conditioning reduces Entropy]
\label{lemma:cond-reduce}
It holds for all $\rv{X}$, $\rv{Y}$ and $E$ that:
\[
\H(\rv{X} \mid \rv{Y}, E) \leq \H(\rv{X} \mid E) .
\]
Equality holds if and only if $\rv{X}$ and $\rv{Y}$ are independent conditioned on $E$.
\end{lemma}
\begin{proof}
We have:
\begin{align*}
\H(\rv{X} \mid \rv{Y}, E) &= \sum_y \Pr(y \mid E) \cdot \H(\rv{X} \mid \rv{Y} = y, E) \\
&= \sum_{x, y} \Pr(y \mid E) \cdot \Pr(x \mid y, E) \cdot \log \frac{1}{\Pr(x \mid y, E)} \\
&= \sum_{x, y} \Pr(x \mid E) \cdot \Pr(y \mid x, E) \cdot \log \frac{\Pr(y \mid E)}{\Pr(x \mid E) \cdot \Pr(y \mid x, E)} \\
&\leq \sum_x \Pr(x \mid E) \cdot \log \frac{1}{\Pr(x \mid E)} \tag{Concavity of $\log(\cdot)$} \\
&= \H(\rv{X} \mid E) . \qedhere
\end{align*}
\end{proof}

\begin{lemma}
\label{lemma:entropy-ub}
It holds for all $\rv{X}$ and $E$ that:
\[
0 \leq \H(\rv{X} \mid E) \leq \log \paren*{ \card*{ \supp(\rv{X}) } } .
\]
The second inequality is tight if and only if $\rv{X}$ conditioned on $E$ is the uniform distribution over $\supp(\rv{X})$.
\end{lemma}
\begin{proof}
The first inequality is direct. For the second, we have by the concavity of $\log(\cdot)$ that:
\[
\H(\rv{X} \mid E) = \sum_x \Pr(x \mid E) \cdot \log \frac{1}{\Pr(x \mid E)} \leq \log \paren*{ \card*{ \supp(\rv{X}) } } . \qedhere 
\]
\end{proof}

\subsection{Min-Entropy}

\begin{definition}[Min-Entropy] 
\label{def:min-entropy}
The min-entropy of a discrete random variable $\rv{X}$ is 
\[
\H_{\infty}(\rv{X}) = \min_{x : \Pr(x) > 0} \log \frac{1}{\Pr(x)}.
\]
\end{definition}

\begin{fact}
\label{fact:min-entropy-ub}
If the random variable $\rv{X}$ takes values in the set $\Omega$, it holds that
\[
0 \leq \H_{\infty}(\rv{X}) \leq \H(\rv{X}) \leq \log \card*{ \Omega }
\]
\end{fact}

Recall that $\h(x) = - x \log(x) - (1-x) \log (1-x)$ for $x \in [0,1]$ is the binary entropy function.

\begin{lemma}
\label{lemma:min-entropy-lb}
If the random variable $\rv{X}$ takes values in the set $\Omega$ and $\card*{\Omega} > 1$, it holds that
\[
2^{ - \H_{\infty}(\rv{X}) } \leq 1 - \frac{ \H(\rv{X}) - 1 }{ \log \paren*{ \card*{ \Omega } } } .
\]
\end{lemma}
\begin{proof}
If $\rv{X}$ is a point mass, there is nothing to show. Otherwise, let $x^*$ be such that $\Pr(x^*)$ is the largest possible (breaking ties arbitrarily). We have:
\begin{align*}
\H(\rv{X}) &= \Pr(x^*) \cdot \log \frac{1}{\Pr(x^*)} + \sum_{x \neq x^*} \Pr(x) \cdot \log \frac{1}{\Pr(x)} \\
&= \h(\Pr(x^*)) - \paren*{ 1 - \Pr(x^*) } \cdot \log \frac{1}{ 1 - \Pr(x^*) } + \sum_{x \neq x^*} \Pr(x) \cdot \log \frac{1}{\Pr(x)} \\
&= \h(\Pr(x^*)) + \paren*{ 1 - \Pr(x^*) } \cdot \sum_{x \neq x^*} \frac{ \Pr(x) }{ 1 - \Pr(x^*) } \cdot \log \frac{ 1 - \Pr(x^*) }{ \Pr(x) } \\
&= \h(\Pr(x^*)) + \paren*{ 1 - \Pr(x^*) } \cdot \H(\rv{X} \mid \rv{X} \neq x^*) .
\end{align*}
Using the fact that $\h(\cdot) \leq 1$ and \autoref{fact:min-entropy-ub}, we have:
\[
\H(\rv{X}) \leq 1 + \paren*{ 1 - \Pr(x^*) } \cdot \log \paren*{ \card*{ \Omega } } .
\]
Rearranging gives:
\[
2^{ - \H_{\infty}(\rv{X}) } = \Pr(x^*) \leq 1 - \frac{ \H(\rv{X}) - 1 }{ \log \paren*{ \card*{ \Omega } } } .
\]
\end{proof}

\subsection{Mutual Information}

\begin{definition}[Mutual Information] 
\label{def:mutualinfo} 
The mutual information between $\rv{X}$ and $\rv{Y}$ is defined as:
\[
\I(\rv{X}:\rv{Y}) = \H(\rv{X}) - \H(\rv{X}\mid \rv{Y}) = \H(\rv{Y}) - \H(\rv{Y}\mid \rv{X}).
\]
The mutual information between $\rv{X}$ and $\rv{Y}$ conditioned on $\rv{Z}$ is defined as:
\[
\I(\rv{X}:\rv{Y} \mid \rv{Z}) = \H(\rv{X}\mid \rv{Z}) - \H(\rv{X}\mid \rv{Y}\rv{Z}) = \H(\rv{Y}\mid \rv{Z}) - \H(\rv{Y}\mid \rv{X}\rv{Z}).
\]
\end{definition}

\begin{fact} 
\label{fact:infbound} 
We have $0 \leq \I(\rv{X} : \rv{Y}\mid \rv{Z})\leq \H(\rv{X})$.
\end{fact}

\begin{fact}[Chain Rule for Mutual Information]  
\label{fact:infchain} 
If $\rv{A}$, $\rv{B}$, $\rv{C}$, $\rv{D}$ are random variables, then
\[
\I(\rv{A}\rv{B} : \rv{C} \mid \rv{D}) =  \I(\rv{A} : \rv{C} \mid \rv{D})+ \I(\rv{B} : \rv{C} \mid \rv{A}\rv{D}).
\]
\end{fact}

The following lemmas are standard.

\begin{lemma}
\label{lemma:addcondition}
For random variables $\rv{A}, \rv{B}, \rv{C}, \rv{D}$, if $\rv{A}$ is independent of $\rv{D}$ given $\rv{C}$, then, 
\[
\I\paren*{ \rv{A} : \rv{B} \mid \rv{C} } \leq \I\paren*{ \rv{A} : \rv{B} \mid \rv{C}, \rv{D} } .
\]
\end{lemma}
\begin{proof}
Since $\rv{A}$ and $\rv{D}$ are independent conditioned on $\rv{C}$, by
\autoref{lemma:cond-reduce}, $\H(\rv{A} \mid  \rv{C}) = \H(\rv{A} \mid \rv{C}, \rv{D})$ and $\H(\rv{A} \mid  \rv{B}, \rv{C}, \rv{D}) \leq \H(\rv{A} \mid  \rv{B}, \rv{C})$.  We have,
\begin{align*}
	\I\paren*{ \rv{A} : \rv{B} \mid \rv{C} } &= \H(\rv{A} \mid \rv{C}) - \H(\rv{A} \mid \rv{B}, \rv{C}) \\
	&= \H(\rv{A} \mid  \rv{C}, \rv{D}) - \H(\rv{A} \mid \rv{B}, \rv{C}) \\
	&\leq \H(\rv{A} \mid \rv{C}, \rv{D}) - \H(\rv{A} \mid \rv{B}, \rv{C}, \rv{D}) = \I\paren*{ \rv{A} : \rv{B} \mid \rv{C}, \rv{D} }.
\end{align*} 
	
\end{proof}

\begin{lemma}
\label{lemma:removecondition}
For random variables $\rv{A}, \rv{B}, \rv{C},\rv{D}$, if $\rv{A}$ is independent of $\rv{D}$ given $\rv{B}, \rv{C}$, then, 
\[
\I\paren*{ \rv{A} : \rv{B} \mid \rv{C}, \rv{D} } \leq \I\paren*{ \rv{A} : \rv{B} \mid \rv{C} } .
\]
\end{lemma}
\begin{proof}
Since $\rv{A}$ and $\rv{D}$ are independent conditioned on $\rv{B}, \rv{C}$, by \autoref{lemma:cond-reduce}, $\H(\rv{A} \mid \rv{B},\rv{C}) = \H(\rv{A} \mid \rv{B},\rv{C},\rv{D})$. Moreover, since conditioning can only reduce the entropy (again by \autoref{lemma:cond-reduce}), 
\begin{align*}
	\I\paren*{ \rv{A} : \rv{B} \mid \rv{C}, \rv{D} } &= \H(\rv{A} \mid \rv{C},\rv{D}) - \H(\rv{A} \mid \rv{B},\rv{C},\rv{D}) \\ 
	&= \H(\rv{A} \mid \rv{C},\rv{D}) - \H(\rv{A} \mid \rv{B},\rv{C}) \\ 
	&\leq \H(\rv{A} \mid \rv{C}) - \H(\rv{A} \mid \rv{B},\rv{C}) = \I\paren*{ \rv{A} : \rv{B} \mid \rv{C} } .
\end{align*}
 
\end{proof}

%% file: files/app-proofs.tex

\section{Missing Proofs in~\autoref{sec:lb}}\label{sec:missing-proofs}

In this section we provide the missing proofs in~\autoref{sec:lb}.

\subsection{Missing Proofs in~\autoref{sec:distprop}}\label{sec:missing-proofs-distprop}

\begin{reminder}{\autoref{lemma:largenbr}}
	For all $k \in [p+1]$, we have:
	\[
	\Pr\paren*{ \card*{ \rv{N}^k(s) } \leq \Delta_k \cdot \paren*{ 1 - \frac{2k}{ \paren*{ \log n }^{ 10^{ 10p } } } } } \leq \frac{k}{n^{200}} .
	\]
\end{reminder}
\begin{proof}
Proof by induction on $k$. The base case $k = 1$ follows from \autoref{obs:smallnbr}. We show the result for $k > 1$ by assuming it holds for $k - 1$. Letting $z_k = \Delta_k \cdot \paren*{ 1 - \frac{2k}{ \paren*{ \log n }^{ 10^{ 10p } } } }$ for convenience, we have:
\begin{align*}
\Pr\paren*{ | \rv{N}^k(s) | \leq z_k } &\leq \Pr\paren*{ | \rv{N}^{k-1}(s) | \leq z_{k-1} } + \Pr\paren*{ | \rv{N}^k(s) | \leq z_k ~\Big|~ | \rv{N}^{k-1}(s) | > z_{k-1} } \\
&\leq \frac{k-1}{n^{200}} + \Pr\paren*{ | \rv{N}^k(s) | \leq z_k ~\Big|~ | \rv{N}^{k-1}(s) | > z_{k-1} } \tag{Induction Hypothesis} .
\end{align*}
Thus, it is sufficient to show that the second term is at most $\frac{1}{n^{200}}$. To show this we fix an set $S$ such that $\card*{S} > z_{k-1}$ and show that:
\begin{equation}
\label{eq:largenbr0}
\Pr\paren*{ | \rv{N}^k(s) | \leq z_k ~\Big|~ \rv{N}^{k-1}(s) = S } \leq \frac{1}{n^{200}} .
\end{equation}
To see why this holds, first note that conditioned on $\rv{N}^{k-1}(s) = S$, the set $\rv{N}^k(s)$ is just the set of vertices in later $V_{k+1}$ that can be reached from vertices in $S$ (which itself is a subset of $V_k$). Thus, we have:
\begin{equation}
\label{eq:largenbr1}
\begin{split}
&\Pr\paren*{ | \rv{N}^k(s) | \leq z_k ~\Big|~ \rv{N}^{k-1}(s) = S } \\
&\hspace{1cm} \leq \Pr\paren*{ \card*{ \set*{ v \in V_{k+1} \mid \exists u \in S \mathrm{~s.t.~} (u,v) \in \rv{E}_k } } \leq z_k ~\Big|~ \rv{N}^{k-1}(s) = S } \\
&\hspace{1cm} = \Pr\paren*{ \card*{ \set*{ v \in V_{k+1} \mid \exists u \in S \mathrm{~s.t.~} (u,v) \in \rv{E}_k } } \leq z_k } ,
\end{split}
\end{equation}
where the last step is because $\rv{N}^{k-1}(s) = S$ is determined by $\rv{E}_{-k}$, which is independent of $\rv{E}_k$. We now want to upper bound the probability of an event defined by $\rv{E}_k$ but instead of analyzing it directly, we first define two auxiliary random variables $\rv{E}'_k$ and $\rv{E}''_k$. The values taken by the random variables $\rv{E}'_k$ and $\rv{E}''_k$ are just a set of edges between $V_k$ and $V_{k+1}$. Let $z$ denote $z = \Delta' \cdot \paren*{ 1 - \frac{1}{ \paren*{ \log n }^{ 10^{ 10p } } } }$. In the random variable $\rv{E}'_k$, each edge $(u,v)$ for $u \in V_k, v \in V_{k+1}$ is included independently with probability $z$. In the random variable $\rv{E}''_k$, we first sample edges as in $\rv{E}'_k$ and then, if the number of edges coming out of any vertex $u \in V_k$ is $d'_u \leq \Delta$, we make it equal to $\Delta$ by sample $\Delta - d'_u$ edges uniformly at random (and do nothing if $d'_u > \Delta$). Denoting by $\mathsf{dist}(\rv{X})$ the distribution of the random variable $\rv{X}$, note first that
\begin{equation}
\label{eq:largenbr2}
\mathsf{dist}\paren*{ \rv{E}''_k \mid \forall u \in V_k \mathrm{~s.t.~} d'_u \leq \Delta } = \mathsf{dist}\paren*{ \rv{E}_k } \implies \tvd{ \mathsf{dist}\paren*{ \rv{E}''_k } }{ \mathsf{dist}\paren*{ \rv{E}_k } } \leq \Pr_{\rv{E}'_k}\paren*{ \exists u \in V_k \mathrm{~s.t.~} d'_u > \Delta } .
\end{equation}
Plugging \autoref{eq:largenbr2} into \autoref{eq:largenbr1} and noting that $\rv{E}'_k$ samples at most as many edges as $\rv{E}''_k$, we have:
\begin{equation}
\label{eq:largenbr3}
\begin{split}
&\Pr\paren*{ | \rv{N}^k(s)| \leq z_k ~\Big|~ \rv{N}^{k-1}(s) = S } \\
&\hspace{1cm} \leq \Pr_{\rv{E}'_k}\paren*{ \card*{ \set*{ v \in V_{k+1} \mid \exists u \in S \mathrm{~s.t.~} (u,v) \in \rv{E}'_k } } \leq z_k } + \Pr_{\rv{E}'_k}\paren*{ \exists u \in V_k \mathrm{~s.t.~} d'_u > \Delta } \\
&\hspace{1cm} \leq \Pr_{\rv{E}'_k}\paren*{ \card*{ \set*{ v \in V_{k+1} \mid \exists u \in S \mathrm{~s.t.~} (u,v) \in \rv{E}'_k } } \leq z_k } + \sum_{u \in V_k} \Pr_{\rv{E}'_k}\paren*{ \card*{ \set*{ v \in V_{k+1} \mid (u,v) \in \rv{E}'_k } } > \Delta } .
\end{split}
\end{equation}
Now, for $v \in V_{k+1}$, define the indicator random variable $\rv{X}_v$ to be $1$ if and only if $\exists u \in S: (u,v) \in \rv{E}'_k$. Also, for $u \in V_k$, define the indicator random variable $\rv{Y}_{u,v}$ to be $1$ if and only if $(u,v) \in \rv{E}'_k$. Clearly, the random variable $\rv{X}_v$ are mutually independent for all $v \in V_{k+1}$ and so are the random variables $\rv{Y}_{u,v}$ for $u \in V_k, v \in V_{k+1}$. Moreover, we have, for all $u \in V_k$ and $v \in V_{k+1}$ that:
\begin{equation}
\label{eq:largenbr4}
 \E\sq*{ \rv{X}_v } = 1 - \paren*{ 1 - z }^{ \card*{ S } } \geq z \cdot z_{k-1} \cdot \paren*{ 1 - \frac{0.5}{ \paren*{ \log n }^{ 10^{10p} } } } \hspace{1cm}\text{and}\hspace{1cm} \E\sq*{ \rv{Y}_{u,v} } = z ,
\end{equation}
as $1 - x \leq \e^{-x} \leq 1 - x + \frac{x^2}{2}, \forall x \in (0,1/10)$ and $z_{k-1} < \card*{S} \leq \Delta_{k-1}$ by \autoref{obs:smallnbr}. We now continue \autoref{eq:largenbr3} using a Chernoff bound (\autoref{lemma:chernoff}).
\begin{align*}
\Pr\paren*{ | \rv{N}^k(s) | \leq z_k ~\Big|~ \rv{N}^{k-1}(s) = S } &\leq \Pr_{\rv{E}'_k}\paren*{ \sum_{v \in V_{k+1}} \rv{X}_v \leq z_k } + \sum_{u \in V_k} \Pr_{\rv{E}'_k}\paren*{ \sum_{v \in V_{k+1}} \rv{Y}_{u,v} > \Delta } \\
&\leq \exp \Big( {- \frac{0.5}{ \paren*{ \log n }^{ 20^{10p} } } \cdot \frac{\Delta}{10} } \Big) + \sum_{u \in V_k} \exp \Big( {-\frac{\Delta}{10} \cdot \frac{1}{ \paren*{ \log n }^{ 20^{ 10p } } } } \Big) \tag{\autoref{lemma:chernoff}} \\
&\leq \frac{1}{n^{200}} ,
\end{align*}
as required for \autoref{eq:largenbr0}.

\end{proof}

\begin{reminder}{\autoref{lemma:entropynbrub}}
	For all $v \in V \setminus \paren*{ V_1 \cup V_{p+2} }$ and any event $E$, we have:
	\begin{align*}
	\H\paren*{ \rv{N}(v) \mid E } &\leq n \cdot \h\paren*{ \Delta' } \cdot \paren*{ 1 + \frac{1}{ \paren*{ \log n }^{ 10^{10p} } } } \text{ and} \\
	\H\paren*{ \rv{N}(v) } &\geq n \cdot \h\paren*{ \Delta' } \cdot \paren*{ 1 - \frac{1}{ \paren*{ \log n }^{ 10^{10p} } } } .
	\end{align*}
\end{reminder}
\begin{proof}
	As $v \notin V_1 \cup V_{p+2}$, we get from \autoref{lemma:binomial} that:
	\[
	\H\paren*{ \rv{N}(v) \mid E } \leq \log \binom{n}{\Delta} \leq \log n + n \cdot \h\paren*{ \Delta' } \leq n \cdot \h\paren*{ \Delta' } \cdot \paren*{ 1 + \frac{1}{ \paren*{ \log n }^{ 10^{10p} } } } .
	\]
	For the furthermore part, note that \autoref{lemma:binomial} also says that:
	\[
	\H\paren*{ \rv{N}(v) \mid E } = \log \binom{n}{\Delta} \geq -3\log n + n \cdot \h\paren*{ \Delta' } \geq n \cdot \h\paren*{ \Delta' } \cdot \paren*{ 1 - \frac{1}{ \paren*{ \log n }^{ 10^{10p} } } } .
	\]
\end{proof}

\begin{reminder}{\autoref{cor:entropynbrub}}
	For all $1 < k \leq p+1$ and any event $E$, we have:
	\begin{align*}
	\H\paren*{ \rv{E}_k \mid E } &\leq n^2 \cdot \h\paren*{ \Delta' } \cdot \paren*{ 1 + \frac{1}{ \paren*{ \log n }^{ 10^{10p} } } } \text{ and}\\
	\H\paren*{ \rv{E}_k } &\geq n^2 \cdot \h\paren*{ \Delta' } \cdot \paren*{ 1 - \frac{1}{ \paren*{ \log n }^{ 10^{10p} } } } .
	\end{align*}
\end{reminder}
\begin{proof}
	From \autoref{lemma:entropy-chain} and \autoref{lemma:cond-reduce}, we have that:
	\[
	\H\paren*{ \rv{E}_k \mid E } \leq \sum_{v \in V_k} \H\paren*{ \rv{N}(v) \mid E } \leq n^2 \cdot \h\paren*{ \Delta' } \cdot \paren*{ 1 + \frac{1}{ \paren*{ \log n }^{ 10^{10p} } } } \tag{\autoref{lemma:entropynbrub}} .
	\]
	From \autoref{lemma:entropy-chain} and independence of $\rv{N}(v)$ for all $v \in V_k$, we have that:
	\[
	\H\paren*{ \rv{E}_k } = \sum_{v \in V_k} \H\paren*{ \rv{N}(v) } \geq n^2 \cdot \h\paren*{ \Delta' } \cdot \paren*{ 1 - \frac{1}{ \paren*{ \log n }^{ 10^{10p} } } } \tag{\autoref{lemma:entropynbrub}} .
	\]
\end{proof}

\begin{reminder}{\autoref{lemma:entropynbrlb}}
	We have $\H\paren*{ \rv{N}(s) } = \log n$ and, for all $1 < k \leq p+1$:
	\[
	\H\paren*{ \rv{N}^k(s) } \geq n \cdot \h\paren*{ \Delta'_k } \cdot \paren*{ 1 - \frac{1}{ \paren*{ \log n }^{ 10^{8p} } } } .
	\]
\end{reminder}
\begin{proof}
	That $\H\paren*{ \rv{N}(s) } = \log n$ follows from \autoref{obs:smallnbr}. For the rest, fix $k > 1$ and define $z_k = \Delta_k \cdot \paren*{ 1 - \frac{1}{ \paren*{ \log n }^{ 10^{9p} } } }$. From \autoref{lemma:cond-reduce}, we can conclude that $\H\paren*{ \rv{N}^k(s) } \geq \H\paren*{ \rv{N}^k(s) ~\Big|~ | \rv{N}^k(s) | }$. By \autoref{def:cond-entropy}, this implies:
	\[
	\H\paren*{ \rv{N}^k(s) } \geq \sum_{z \geq z_k} \Pr\paren*{ | \rv{N}^k(s) | = z } \cdot \H\paren*{ \rv{N}^k(s) ~\Big|~ | \rv{N}^k(s) | = z } .
	\]
	By symmetry, conditioned $\card*{ \rv{N}^k(s) } = z$, $\rv{N}^k(s)$ is just a uniformly random subset of size $z$. Thus, we have from \autoref{lemma:entropy-ub} that:
	\begin{align*}
	\H\paren*{ \rv{N}^k(s) } &\geq \sum_{z \geq z_k} \Pr\paren*{ | \rv{N}^k(s) | = z } \cdot \log \binom{n}{z} \\
	&\geq \sum_{z \geq z_k} \Pr\paren*{ | \rv{N}^k(s) | = z } \cdot \paren*{ n \cdot \h(z/n) - 3 \log n } \tag{\autoref{lemma:binomial}} \\
	&\geq \paren*{ n \cdot \h(z_k/n) - 3 \log n } \cdot \sum_{z \geq z_k} \Pr\paren*{ | \rv{N}^k(s) | = z } \tag{As $\h(\cdot)$ is increasing on $0 < x < \frac{1}{2}$} \\
	&\geq \paren*{ n \cdot \h\paren*{ \Delta'_k \cdot \paren*{ 1 - \frac{1}{ \paren*{ \log n }^{ 10^{9p} } } } } - 3 \log n } \cdot \sum_{z \geq z_k} \Pr\paren*{ | \rv{N}^k(s) | = z } \tag{Definition of $z_k$} .
	\end{align*}
	As $k > 1$,  the first factor is non-negative. Bounding the second by \autoref{cor:nbr}, we have:
	\begin{align*}
	\H\paren*{ \rv{N}^k(s) } &\geq \paren*{ n \cdot \h\paren*{ \Delta'_k \cdot \paren*{ 1 - \frac{1}{ \paren*{ \log n }^{ 10^{9p} } } } } - 3 \log n } \cdot \paren*{ 1 - \frac{1}{n^{150}} } \\
	&\geq \paren*{ n \cdot \paren*{ 1 - \frac{1}{ \paren*{ \log n }^{ 10^{9p} } } } \cdot \h\paren*{ \Delta'_k } - 3 \log n } \cdot \paren*{ 1 - \frac{1}{n^{150}} } \tag{Concavity of $\h(\cdot)$} \\
	&\geq n \cdot \h\paren*{ \Delta'_k } \cdot \paren*{ 1 - \frac{1}{ \paren*{ \log n }^{ 10^{8p} } } } \tag{As $k > 1$} .
	\end{align*}
\end{proof}

\begin{reminder}{\autoref{lemma:entropypath}}
	For all events $E$, we have $\H\paren*{ \rv{P}(s) \mid E } \leq \log n$ and, for all $1 < k \leq p+1$:
	\[
	\H\paren*{ \rv{P}^k(s) ~\Big|~ E } \leq \Delta'_{k-1} \cdot \paren*{ 1 + \frac{1}{ \paren*{ \log n }^{ 10^{9p} } } } \cdot \H\paren*{ \rv{E}_k }.
	\]
\end{reminder}
\begin{proof}
	Proof by induction. The base case $k = 1$ follows from \autoref{obs:smallnbr}. We show the result for $k > 1$ assuming it holds for $k-1$. As $\rv{P}^k(s)$ is determined by $\rv{P}^{k-1}(s)$ and $\rv{P}(\rv{N}^{k-1}(s))$, we have:
	\begin{align*}
	\H\paren*{ \rv{P}^k(s) ~\Big|~ E } &\leq \H\paren*{ \rv{P}^{k-1}(s) \rv{P}(\rv{N}^{k-1}(s)) ~\Big|~ E } \\
	&\leq \H\paren*{ \rv{P}^{k-1}(s) ~\Big|~ E } + \H\paren*{ \rv{P}(\rv{N}^{k-1}(s)) ~\Big|~ \rv{P}^{k-1}(s), E } \tag{\autoref{lemma:entropy-chain}} \\
	&\leq \H\paren*{ \rv{P}^{k-1}(s) \mid E } + \sum_{P^{k-1}(s)} \Pr\paren*{ P^{k-1}(s) ~\Big|~ E } \cdot \H\paren*{ \rv{P}(\rv{N}^{k-1}(s)) ~\Big|~ P^{k-1}(s), E } \tag{\autoref{def:cond-entropy}} \\
	&\leq \H\paren*{ \rv{P}^{k-1}(s) \mid E } + \sum_{P^{k-1}(s)} \Pr\paren*{ P^{k-1}(s) ~\Big|~ E } \cdot \H\paren*{ \rv{P}(N^{k-1}(s)) ~\Big|~ P^{k-1}(s), E } \tag{As $\rv{P}^{k-1}(s)$ determines $\rv{N}^{k-1}(s)$} .
	\end{align*}
	To continue, we again use \autoref{lemma:entropy-chain} followed by \autoref{lemma:cond-reduce}. 
	\begin{align*}
	\H\paren*{ \rv{P}^k(s) ~\Big|~ E } &\leq \H\paren*{ \rv{P}^{k-1}(s) ~\Big|~ E } + \sum_{P^{k-1}(s)} \Pr\paren*{ P^{k-1}(s) ~\Big|~ E } \cdot \sum_{ v \in N^{k-1}(s) } \H\paren*{ \rv{P}(v) ~\Big|~ P^{k-1}(s), E } \\
	&\leq \H\paren*{ \rv{P}^{k-1}(s) ~\Big|~ E } + \sum_{P^{k-1}(s)} \Pr\paren*{ P^{k-1}(s) ~\Big|~ E } \cdot | N^{k-1}(s) | \cdot n \h\paren*{ \Delta' } \paren*{ 1 + \frac{1}{ \paren*{ \log n }^{ 10^{10p} } } } \tag{\autoref{lemma:entropynbrub}} \\
	&\leq \H\paren*{ \rv{P}^{k-1}(s) ~\Big|~ E } + \sum_{P^{k-1}(s)} \Pr\paren*{ P^{k-1}(s) ~\Big|~ E } \cdot \Delta_{k-1} \cdot n \h\paren*{ \Delta' } \paren*{ 1 + \frac{1}{ \paren*{ \log n }^{ 10^{10p} } } } \tag{\autoref{obs:smallnbr}} \\
	&\leq \H\paren*{ \rv{P}^{k-1}(s) ~\Big|~ E } + \Delta_{k-1} \cdot n \h\paren*{ \Delta' } \paren*{ 1 + \frac{1}{ \paren*{ \log n }^{ 10^{10p} } } } \\
	&\leq \H\paren*{ \rv{P}^{k-1}(s) ~\Big|~ E } + \Delta'_{k-1} \cdot \H\paren*{ \rv{E}_k } \cdot \paren*{ 1 + \frac{5}{ \paren*{ \log n }^{ 10^{10p} } } } \tag{\autoref{cor:entropynbrub}} .
	\end{align*}

	Finally, we bound the term $\H\paren*{ \rv{P}^{k-1}(s) \mid E }$ using the induction hypothesis. When $k = 2$, this term is at most $\log n \leq \frac{ \H\paren*{ \rv{E}_k } }{ n \cdot \paren*{ \log n }^{ 10^{10p} } } = \frac{ \Delta'_{k-1} }{ \paren*{ \log n }^{ 10^{10p} } } \cdot \H\paren*{ \rv{E}_k }$. Otherwise, when $k > 2$, we have $\H\paren*{ \rv{E}_{k-1} } = \H\paren*{ \rv{E}_k }$ and this term is at most $\Delta'_{k-2} \cdot \paren*{ 1 + \frac{1}{ \paren*{ \log n }^{ 10^{9p} } } } \cdot \H\paren*{ \rv{E}_k } \leq \frac{ \Delta'_{k-1} }{ \paren*{ \log n }^{ 10^{10p} } } \cdot \H\paren*{ \rv{E}_k }$. Plugging in, we have:
	\[
	\H\paren*{ \rv{P}^k(s) } \leq \Delta'_{k-1} \cdot \paren*{ 1 + \frac{1}{ \paren*{ \log n }^{ 10^{9p} } } } \cdot \H\paren*{ \rv{E}_k } .
	\]

\end{proof}

\begin{reminder}{\autoref{lemma:entropypathlb}}
We have $\H\paren*{ \rv{P}(s) } = \log n$ and, for all $1 < k \leq p+1$:
\[
\H\paren*{ \rv{P}^k(s) } \geq \Delta'_{k-1} \cdot \paren*{ 1 - \frac{1}{ \paren*{ \log n }^{ 10^{9p} } } } \cdot \H\paren*{ \rv{E}_k }.
\]
\end{reminder}
\begin{proof}
	The case $k = 1$ follows from \autoref{obs:smallnbr}. We show the result for $k > 1$. As $\rv{P}^k(s)$ determines $\rv{P}(\rv{N}^{k-1}(s))$, we have:
	\begin{align*}
	\H\paren*{ \rv{P}^k(s) } &\geq \H\paren*{ \rv{P}(\rv{N}^{k-1}(s)) } \\
	&\geq \H\paren*{ \rv{P}(\rv{N}^{k-1}(s)) ~\Big|~ \rv{P}^{k-1}(s) } \tag{\autoref{lemma:cond-reduce}} \\
	&\geq \sum_{P^{k-1}(s)} \Pr\paren*{ P^{k-1}(s) } \cdot \H\paren*{ \rv{P}(\rv{N}^{k-1}(s)) ~\Big|~ P^{k-1}(s) } \tag{\autoref{def:cond-entropy}} \\
	&\geq \sum_{P^{k-1}(s)} \Pr\paren*{ P^{k-1}(s) } \cdot \H\paren*{ \rv{P}(N^{k-1}(s)) ~\Big|~ P^{k-1}(s) } \tag{As $\rv{P}^{k-1}(s)$ determines $\rv{N}^{k-1}(s)$} .
	\end{align*}
	Note that $\rv{P}(N^{k-1}(s))$ is determined by $\rv{E}_k$ and $\rv{P}^{k-1}(s)$ is determined by $\rv{E}_{-k}$. As these are independent, we have:
	\[
	\H\paren*{ \rv{P}^k(s) } \geq \sum_{P^{k-1}(s)} \Pr\paren*{ P^{k-1}(s) } \cdot \H\paren*{ \rv{P}(N^{k-1}(s)) } .
	\]
	As $\rv{P}(v)$ are mutually independent for all $v \in N^{k-1}(s)$, we have:
	\begin{align*}
	\H\paren*{ \rv{P}^k(s) } &\geq \sum_{P^{k-1}(s)} \Pr\paren*{ P^{k-1}(s) } \cdot \sum_{ v \in N^{k-1}(s) } \H\paren*{ \rv{P}(v) } \\
	&\geq \sum_{P^{k-1}(s)} \Pr\paren*{ P^{k-1}(s) } \cdot \card*{ N^{k-1}(s) } \cdot \frac{ \H\paren*{ \rv{E}_k } }{n} .
	\end{align*}
	Using \autoref{lemma:largenbr}, we get:
	\begin{align*}
	\H\paren*{ \rv{P}^k(s) } &\geq \Delta_{k-1} \cdot \paren*{ 1 - \frac{1}{ \paren*{ \log n }^{ 10^{ 9p } } } } \cdot \frac{ \H\paren*{ \rv{E}_k } }{n} \\
	&\geq \Delta'_{k-1} \cdot \paren*{ 1 - \frac{1}{ \paren*{ \log n }^{ 10^{9p} } } } \cdot \H\paren*{ \rv{E}_k } .
	\end{align*}
\end{proof}

\begin{reminder}{\autoref{lemma:minentropypath}}
For all events $E$, it holds that:
\[
2^{ - \H_{\infty}\paren*{ \rv{P}^{p+1}(s) \mid E } } \leq 1 - \frac{ \H\paren*{ \rv{P}^{p+1}(s) \mid E } - 1 }{ \Delta'_p \cdot \paren*{ 1 + \frac{1}{ \paren*{ \log n }^{ 10^{9p} } } } \cdot \H\paren*{ \rv{E}_{p+1} } } .
\]
\end{reminder}
\begin{proof}
Let $\Omega$ denote the support of $\rv{P}^{p+1}(s)$ and note that \autoref{lemma:entropypath} implies that $\log \paren*{ \card*{ \Omega } } \leq \Delta'_p \cdot \paren*{ 1 + \epsilon_p } \cdot \H\paren*{ \rv{E}_{p+1} }$. Applying \autoref{lemma:min-entropy-lb} on the random variable $\rv{P}^{p+1}(s) \mid E$, we have:
\begin{align*}
2^{ - \H_{\infty}\paren*{ \rv{P}^{p+1}(s) \mid E } } &\leq 1 - \frac{ \H\paren*{ \rv{P}^{p+1}(s) \mid E } - 1 }{ \log \paren*{ \card*{ \Omega } } } \\
&\leq 1 - \frac{ \H\paren*{ \rv{P}^{p+1}(s) \mid E } - 1 }{ \Delta'_p \cdot \paren*{ 1 + \epsilon_p } \cdot \H\paren*{ \rv{E}_{p+1} } } .
\end{align*}
\end{proof}

\subsection{Missing Proofs in~\autoref{sec:cclb}}\label{sec:missing-proofs-cclb}

\subsubsection{Proof of~\autoref{lemma:bad1}}

\begin{reminder}{\autoref{lemma:bad1}}
	For all $1 < k \leq p + 1$, assuming~\autoref{lemma:induction} holds for $k-1$, there exists a set $\B_1 \subseteq \supp\paren*{ \rv{M}_{\leq (k-1)p} }$ such that $\Pr\paren*{ \rv{M}_{\leq (k-1)p} \in \B_1 } \leq \sqrt{ \eps_{k-1} }$ and for all $M_{\leq (k-1)p} \notin \B_1$, we have:
	\begin{align*}
	&\H\paren*{ \rv{N}^{k-1}(s) \mid M_{\leq (k-1)p} } \geq \H\paren*{ \rv{N}^{k-1}(s) } - 10 \cdot \sqrt{ \eps_{k-1} } \cdot \Delta'_{k-2} \cdot \H\paren*{ \rv{E}_{k-1} } , &\text{~if~} k > 2 \\
	&\H\paren*{ \rv{N}^{k-1}(s) \mid M_{\leq (k-1)p} } = \H\paren*{ \rv{N}^{k-1}(s) } = \log n, &\text{~if~} k = 2.
	\end{align*}
\end{reminder}
\vspace{-\bigskipamount}
\vspace{-\bigskipamount}
\begin{proof}
	If $k = 2$, then we have from the assumption that \autoref{lemma:induction} holds for $k-1$ that $\H\paren*{ \rv{P}(s) \mid \rv{M}_{\leq t'} } = \H\paren*{ \rv{P}(s) } = \log n$. Thus, we have $\I\paren*{ \rv{N}(s) : \rv{M}_{\leq t'} } \leq \I\paren*{ \rv{P}(s) : \rv{M}_{\leq t'} } = 0$ and the result follows.
	
	If $k > 2$, we have from the assumption that \autoref{lemma:induction} holds for $k-1$ that $\H\paren*{ \rv{P}^{k-1}(s) \mid \rv{M}_{\leq t'} } \geq \paren*{ 1 - \epsilon_{k-1} } \cdot \Delta'_{k-2} \cdot \H\paren*{ \rv{E}_{k-1} }$. Combining with \autoref{lemma:entropypath}, we have that:
	\[
	\I\paren*{ \rv{N}^{k-1}(s) : \rv{M}_{\leq t'} } \leq \I\paren*{ \rv{P}^{k-1}(s) : \rv{M}_{\leq t'} } \leq 2 \epsilon_{k-1} \cdot \Delta'_{k-2} \cdot \H\paren*{ \rv{E}_{k-1} } .
	\]
	Using \autoref{def:mutualinfo} and \autoref{def:cond-entropy}, we get that:
	\begin{equation}
	\label{eq:bad1:1}
	\E_{ M_{\leq t'} }\sq*{ \H\paren*{ \rv{N}^{k-1}(s) } - \H\paren*{ \rv{N}^{k-1}(s) \mid M_{\leq t'} } } \leq 2 \epsilon_{k-1} \cdot \Delta'_{k-2} \cdot \H\paren*{ \rv{E}_{k-1} } .
	\end{equation}
	
	We claim that:
	\begin{claim}
		\label{claim:bad1:1}
		It holds for all $M_{\leq t'}$ that:
		\[
		\H\paren*{ \rv{N}^{k-1}(s) } - \H\paren*{ \rv{N}^{k-1}(s) \mid M_{\leq t'} } \geq - 4 \epsilon_{k-1} \cdot \Delta'_{k-2} \cdot \H\paren*{ \rv{E}_{k-1} } .
		\]
	\end{claim}
	\begin{proof}
		We derive:
		\begin{align*}
		\H\paren*{ \rv{N}^{k-1}(s) ~\Big|~ M_{\leq t'} } &\leq \H\paren*{ | \rv{N}^{k-1}(s) | ~\Big|~ M_{\leq t'} } + \H\paren*{ \rv{N}^{k-1}(s) ~\Big|~ | \rv{N}^{k-1}(s) |, M_{\leq t'} } \tag{\autoref{lemma:entropy-chain}} \\
		&\leq \log n + \log \binom{n}{\Delta_{k-1}} \tag{\autoref{lemma:entropy-ub} and $\card*{ \rv{N}^{k-1}(s) } \leq \Delta_{k-1}$ by \autoref{obs:smallnbr}} \\
		&\leq 2 \log n + n \cdot \h\paren*{ \Delta'_{k-1} } \tag{\autoref{lemma:binomial}} \\
		&\leq \H\paren*{ \rv{N}^{k-1}(s) } \cdot \paren*{ 1 + 2 \epsilon_{k-1} } \tag{\autoref{lemma:entropynbrlb}} \\
		&\leq \H\paren*{ \rv{N}^{k-1}(s) } + \H\paren*{ \rv{P}^{k-1}(s) } \cdot 2 \epsilon_{k-1} \tag{$\H\paren*{ \rv{N}^{k-1}(s) } \le \H\paren*{ \rv{P}^{k-1}(s) }$}\\
		&\leq \H\paren*{ \rv{N}^{k-1}(s) } + 4 \epsilon_{k-1} \cdot \Delta'_{k-2} \cdot \H\paren*{ \rv{E}_{k-1} } \tag{\autoref{lemma:entropypath}} .
		\end{align*}
	\end{proof}

	Conclude from \autoref{eq:bad1:1} that:
	\[
	\E_{ M_{\leq t'} }\sq*{ \H\paren*{ \rv{N}^{k-1}(s) } - \H\paren*{ \rv{N}^{k-1}(s) ~\Big|~ M_{\leq t'} } + 4 \epsilon_{k-1} \cdot \Delta'_{k-2} \cdot \H\paren*{ \rv{E}_{k-1} } } \leq 10 \epsilon_{k-1} \cdot \Delta'_{k-2} \cdot \H\paren*{ \rv{E}_{k-1} } .
	\]
	As the left hand side above is always non-negative by \autoref{claim:bad1:1}, we can apply Markov's inequality to conclude that:
	\[
	\Pr_{ M_{\leq t'} }\paren*{ \H\paren*{ \rv{N}^{k-1}(s) } - \H\paren*{ \rv{N}^{k-1}(s) ~\Big|~ M_{\leq t'} } \geq 10 \cdot \sqrt{ \epsilon_{k-1} } \cdot \Delta'_{k-2} \cdot \H\paren*{ \rv{E}_{k-1} } } \leq \sqrt{ \epsilon_{k-1} } . 
	\]

\end{proof}

\subsubsection{Proof of~\autoref{lemma:size}}

\begin{reminder}{\autoref{lemma:size}}
	For all $1 < k \leq p + 1$, assuming~\autoref{lemma:induction} holds for $k-1$, and let $\B_1$ be the set promised by \autoref{lemma:bad1}. For all $M_{\leq t'} \notin \B_1$, we have 
	\[
	\card*{ \set*{ i \in [n] \mid \Pr\paren*{ v^{(i)}_k \in \rv{N}^{k-1}(s) ~\Big|~ M_{\leq t'} } \leq \frac{ \Delta'_{k-1} }{ 1 + \eps_k^5 } } } \leq \eps_k^5 \cdot n.
	\]
\end{reminder}
\begin{proof}
	Fix $M_{\leq t'} \notin \B_1$ and let 
	\begin{align*}
	S = \set*{ i \in [n] ~\Big|~ \Pr\paren*{ v^{(i)}_k \in \rv{N}^{k-1}(s) \mid M_{\leq t'} } \leq \frac{ \Delta'_{k-1} }{ 1 + \eps_k^5 } }
	\end{align*}
	for convenience. For $k = 2$, our definitions imply $S = \emptyset$, so we assume $k > 2$ and proceed by contradiction. Suppose that $\card*{S} > \epsilon_k^5 \cdot n$.  For $i \in [n]$, define the indicator random variable $\rv{X}_i$ to be $1$ if and only if $v^{(i)}_k \in \rv{N}^{k-1}(s)$ and define $\alpha_i = \Pr\paren*{ v^{(i)}_k \in \rv{N}^{k-1}(s) \mid M_{\leq t'} } = \Pr\paren*{ \rv{X}_i = 1 \mid M_{\leq t'} }$. By \autoref{lemma:bad1}, we have:
	\begin{align*}
	\H\paren*{ \rv{X}_1 \cdots \rv{X}_n \mid M_{\leq t'} } &\geq \H\paren*{ \rv{N}^{k-1}(s) \mid M_{\leq t'} } \\
	&\geq \H\paren*{ \rv{N}^{k-1}(s) } - 10 \cdot \sqrt{ \epsilon_{k-1} } \cdot \Delta'_{k-2} \cdot \H\paren*{ \rv{E}_{k-1} } \\
	&\geq n \cdot \h\paren*{ \Delta'_{k-1} } \cdot \paren*{ 1 - \epsilon_{k-1} } - 10 \cdot \sqrt{ \epsilon_{k-1} } \cdot \Delta'_{k-2} \cdot \H\paren*{ \rv{E}_{k-1} } \tag{\autoref{lemma:entropynbrlb}} .
	\end{align*}
	
	By \autoref{lemma:entropy-chain} and \autoref{lemma:cond-reduce}, we also have $\H\paren*{ \rv{X}_1 \cdots \rv{X}_n \mid M_{\leq t'} } \leq \sum_{i = 1}^n \h\paren*{ \alpha_i }$ implying that:
	\begin{equation}
	\label{eq:size:1}
	n \cdot \h\paren*{ \Delta'_{k-1} } \cdot \paren*{ 1 - \epsilon_{k-1} } - 10 \cdot \sqrt{ \epsilon_{k-1} } \cdot \Delta'_{k-2} \cdot \H\paren*{ \rv{E}_{k-1} } \leq \sum_{i = 1}^n \h\paren*{ \alpha_i } .
	\end{equation}
	
	Using the notation $\Delta''_{k-1} = \frac{ \Delta'_{k-1} }{ 1 + \epsilon_k^5 }$, we use the following claim to upper bound $\sum_{i = 1}^n \h\paren*{ \alpha_i }$. 
	\begin{claim}
		\label{claim:size:1}
		It holds that:
		\[
		\sum_{i = 1}^n \h\paren*{ \alpha_i } \leq \epsilon_k^5 \cdot n \cdot \h\paren*{ \Delta''_{k-1} } + n \cdot \paren*{ 1 - \epsilon_k^5 } \cdot \h\paren*{ \frac{ \Delta''_{k-1} }{ 1 - \epsilon_k^5 } } .
		\]
	\end{claim}
	\begin{proof}
		We break the proof into two cases. The easy case is when $\sum_{i = 1}^n \alpha_i \leq n \Delta''_{k-1}$. In this case, we simply use the concavity and monotonicity of $\h(\cdot)$ on $0 < x < \frac{1}{2}$ to get:
		\[
		\sum_{i = 1}^n \h\paren*{ \alpha_i } \leq n \cdot \h\paren*{ \Delta''_{k-1} } \leq \epsilon_k^5 \cdot n \cdot \h\paren*{ \Delta''_{k-1} } + n \cdot \paren*{ 1 - \epsilon_k^5 } \cdot \h\paren*{ \frac{ \Delta''_{k-1} }{ 1 - \epsilon_k^5 } } ,
		\]
		and the claim follows. We now deal with the hard case $\sum_{i = 1}^n \alpha_i > n \Delta''_{k-1}$. In this case, by the definition of $S$, there exists a $\lambda \in [0,1]$ satisfying $\frac{ \sum_{i \in S} \alpha_i + \lambda \cdot \sum_{i \notin S} \alpha_i  }{ \card*{S} + \lambda \cdot \card*{\overline{S}} } = \Delta''_{k-1}$. This is equivalent to:
		\begin{equation}
		\label{eq:lambda}
		\sum_{i = 1}^n \alpha_i - \Delta''_{k-1} \cdot \card*{S} =  \paren*{ 1 - \lambda } \cdot \sum_{i \notin S} \alpha_i + \lambda \cdot \Delta''_{k-1} \cdot \card*{\overline{S}} .
		\end{equation}
		Using the concavity of $\h$ multiple times, we have:
		\begin{align*}
		\sum_{i = 1}^n \h\paren*{ \alpha_i } &\leq \sum_{i \in S} \h\paren*{ \alpha_i } + \lambda \cdot \sum_{i \notin S} \h\paren*{ \alpha_i } + \paren*{ 1- \lambda } \cdot \sum_{i \notin S} \h\paren*{ \alpha_i } \\
		&\leq \paren*{ \card*{S} + \lambda \cdot \card*{\overline{S}} } \cdot \h\paren*{ \Delta''_{k-1} } + \paren*{ 1- \lambda } \cdot \card*{\overline{S}} \cdot \h\paren*{ \frac{ \sum_{i \notin S} \alpha_i }{ \card*{\overline{S}} } } \\
		&\leq \card*{S} \cdot \h\paren*{ \Delta''_{k-1} } + \card*{\overline{S}} \cdot \h\paren*{ \frac{ \sum_{i = 1}^n \alpha_i - \Delta''_{k-1} \cdot \card*{S} }{ \card*{\overline{S}} } } \tag{\autoref{eq:lambda}} \\
		&\leq \epsilon_k^5 \cdot n \cdot \h\paren*{ \Delta''_{k-1} } + n \cdot \paren*{ 1 - \epsilon_k^5 } \cdot \h\paren*{ \frac{ \sum_{i = 1}^n \alpha_i - n \cdot \epsilon_k^5 \cdot \Delta''_{k-1} }{ n \cdot \paren*{ 1 - \epsilon_k^5 } } } \tag{As $\card*{S} > \epsilon_k^5 \cdot n$} .
		\end{align*}
		To continue, note by \autoref{obs:smallnbr} that $\sum_{i = 1}^n \alpha_i = \E\sq*{ \card*{ \rv{N}^{k-1}(s) } \mid M_{\leq t'} } \leq \Delta_{k-1}$. This gives:
		\[
		\sum_{i = 1}^n \h\paren*{ \alpha_i } \leq \epsilon_k^5 \cdot n \cdot \h\paren*{ \Delta''_{k-1} } + n \cdot \paren*{ 1 - \epsilon_k^5 } \cdot \h\paren*{ \frac{ \Delta''_{k-1} }{ 1 - \epsilon_k^5 } } .
		\]
		
	\end{proof}

	Combining \autoref{claim:size:1} and \autoref{eq:size:1} and rearranging, we get:
	\begin{multline}
	\label{eq:size:2}
	\h\paren*{ \Delta'_{k-1} } - \epsilon_k^5 \cdot \h\paren*{ \Delta''_{k-1} } - \paren*{ 1 - \epsilon_k^5 } \cdot \h\paren*{ \frac{ \Delta''_{k-1} }{ 1 - \epsilon_k^5 } }  \\
	\leq \epsilon_{k-1} \cdot \h\paren*{ \Delta'_{k-1} } + \frac{10}{n} \cdot \sqrt{ \epsilon_{k-1} } \cdot \Delta'_{k-2} \cdot \H\paren*{ \rv{E}_{k-1} } .
	\end{multline}
	
	To derive a contradiction, we show that \autoref{eq:size:2} cannot hold. For this, we first lower bound the left hand side. Recall that $\h(x) = - x \log(x) - (1-x) \log (1-x)$ and observe that both these terms are concave. Thus, we can lower bound:
	\begin{align*}
	&\h\paren*{ \Delta'_{k-1} } - \epsilon_k^5 \cdot \h\paren*{ \Delta''_{k-1} } - \paren*{ 1 - \epsilon_k^5 } \cdot \h\paren*{ \frac{ \Delta''_{k-1} }{ 1 - \epsilon_k^5 } } \\
	&\hspace{0.5cm}\geq \Delta'_{k-1} \cdot \log \frac{1}{ \Delta'_{k-1} } - \epsilon_k^5 \cdot \Delta''_{k-1} \cdot \log \frac{1}{ \Delta''_{k-1} } - \Delta''_{k-1} \cdot \log \frac{ 1 - \epsilon_k^5 }{ \Delta''_{k-1} } \\
	&\hspace{0.5cm}\geq \Delta''_{k-1} \cdot \paren*{ \epsilon_k^5 \cdot \log \frac{1}{ 1 + \epsilon_k^5 } + \log \frac{1}{ 1 - \epsilon_k^{10} } } \tag{As $\Delta''_{k-1} = \frac{ \Delta'_{k-1} }{ 1 + \epsilon_k^5 }$} \\
	&\hspace{0.5cm}\geq \log \e \cdot \Delta''_{k-1} \cdot \paren*{\frac{ \epsilon_k^{15} }{2} + \frac{ \epsilon_k^{20} }{6} } \tag{As $\log (1 + x) \leq \log \e \cdot ( x - x^2/2 + x^3/3 )$ and $\log \frac{1}{1 - x} \geq \log \e \cdot ( x + x^2/2)$} .
	\end{align*}
	Simplifying, we get:
	\begin{equation}
	\label{eq:size:3}
	\h\paren*{ \Delta'_{k-1} } - \epsilon_k^5 \cdot \h\paren*{ \Delta''_{k-1} } - \paren*{ 1 - \epsilon_k^5 } \cdot \h\paren*{ \frac{ \Delta''_{k-1} }{ 1 - \epsilon_k^5 } } \geq \Delta'_{k-1} \cdot \epsilon_k^{25} .
	\end{equation}
	We now upper bound the right hand side of \autoref{eq:size:2}. 
	\begin{align*}
	&\epsilon_{k-1} \cdot \h\paren*{ \Delta'_{k-1} } + \frac{10}{n} \cdot \sqrt{ \epsilon_{k-1} } \cdot \Delta'_{k-2} \cdot \H\paren*{ \rv{E}_{k-1} } \\
	&\hspace{0.5cm}\leq \epsilon_k^{100} \cdot \h\paren*{ \Delta'_{k-1} } + \frac{10}{n} \cdot \epsilon_k^{50} \cdot \Delta'_{k-2} \cdot \H\paren*{ \rv{E}_k } \tag{Definition of $\epsilon_k$ and $k > 2$} \\
	&\hspace{0.5cm}\leq \epsilon_k^{100} \cdot \h\paren*{ \Delta'_{k-1} } + \epsilon_k^{40} \cdot \Delta_{k-2} \cdot \h\paren*{ \Delta' } \tag{\autoref{cor:entropynbrub}} .
	\end{align*}
	Now, note that, for $\frac{1}{n} < x < \frac{1}{2}$, we have $\h(x) \leq -2x \log x \leq 2x \log n$. We get:
	\begin{equation}
	\label{eq:size:4}
	\epsilon_{k-1} \cdot \h\paren*{ \Delta'_{k-1} } + \frac{10}{n} \cdot \sqrt{ \epsilon_{k-1} } \cdot \Delta'_{k-2} \cdot \H\paren*{ \rv{E}_{k-1} } \leq 2 \log n \cdot \epsilon_k^{30} \cdot \Delta'_{k-1} .
	\end{equation}
	
	\autoref{eq:size:2}, \autoref{eq:size:3}, and \autoref{eq:size:4} cannot all hold together, a contradiction.
	
\end{proof}

\begin{reminder}{\autoref{lemma:bad2}}
	There exists a set $\B^* \subseteq \supp\paren*{ \rv{M}_{\leq T} }$ such that $\Pr\paren*{ \rv{M}_{\leq T} \in \B^* } \leq \sqrt{ \eps_{p+1} }$ and for all $M_{\leq T} \notin \B^*$, we have:
	\[
	\H\paren*{ \rv{P}^{p+1}(s) \mid M_{\leq T} } \geq \H\paren*{ \rv{P}^{p+1}(s) } - 10 \cdot \sqrt{ \eps_{p+1} } \cdot \Delta'_p \cdot \H\paren*{ \rv{E}_{p+1} } .
	\]
\end{reminder}

\begin{proof}
As $\H\paren*{ \rv{P}^{p+1}(s) \mid \rv{M}_{\leq T} } \geq \paren*{ 1 - \eps_{p+1} } \cdot \Delta'_p \cdot \H\paren*{ \rv{E}_{p+1} }$, we can conclude from \autoref{lemma:entropypath} that:
\[
\I\paren*{ \rv{P}^{p+1}(s) : \rv{M}_{\leq T} } \leq 2 \cdot \eps_{p+1} \cdot \Delta'_p \cdot \H\paren*{ \rv{E}_{p+1} } .
\]
Using \autoref{def:mutualinfo} and \autoref{def:cond-entropy}, we get that:
\begin{equation}
\label{eq:bad2:1}
\E_{M_{\leq T}}\sq*{ \H\paren*{ \rv{P}^{p+1}(s) } - \H\paren*{ \rv{P}^{p+1}(s) \mid M_{\leq T} } } \leq 2 \cdot \eps_{p+1} \cdot \Delta'_p \cdot \H\paren*{ \rv{E}_{p+1} } .
\end{equation}
Using \autoref{lemma:entropypath} and \autoref{lemma:entropypathlb}, we get that, for all $M_{\leq T}$:
\begin{equation}
\label{eq:bad2:2}
\H\paren*{ \rv{P}^{p+1}(s) } - \H\paren*{ \rv{P}^{p+1}(s) \mid M_{\leq T} } \geq - 2 \cdot \eps_{p+1} \cdot \Delta'_p \cdot \H\paren*{ \rv{E}_{p+1} } .
\end{equation}
Conclude from \autoref{eq:bad2:1} that:
\[
\E_{M_{\leq T}}\sq*{ \H\paren*{ \rv{P}^{p+1}(s) } - \H\paren*{ \rv{P}^{p+1}(s) \mid M_{\leq T} } + 2 \cdot \eps_{p+1} \cdot \Delta'_p \cdot \H\paren*{ \rv{E}_{p+1} } } \leq 4 \cdot \eps_{p+1} \cdot \Delta'_p \cdot \H\paren*{ \rv{E}_{p+1} } .
\]
As the left hand side above is always non-negative by \autoref{eq:bad2:2}, we can apply Markov's inequality to conclude that:
\[
\Pr_{M_{\leq T}}\sq*{ \H\paren*{ \rv{P}^{p+1}(s) } - \H\paren*{ \rv{P}^{p+1}(s) \mid M_{\leq T} } \geq 10 \cdot \sqrt{ \eps_{p+1} } \cdot \Delta'_p \cdot \H\paren*{ \rv{E}_{p+1} } } \leq \sqrt{ \eps_{p+1} } .
\]
The lemma follows.

\end{proof}

%% file: files/app-starting-vertex-oblivious-lowb.tex

\section{Lower Bounds against Starting Vertex Oblivious Streaming Algorithms}\label{sec:lowb-input-oblivious}

In this section we prove~\autoref{theo:lowb-main-oblivious-algo} (restated below). See also~\autoref{defi:starting-vertex-oblivious} for a formal definition of starting vertex oblivious streaming algorithm for simulating random walks.

\begin{reminder}{\autoref{theo:lowb-main-oblivious-algo}}
	Let $n \geq 1$ be a sufficiently large integer and let $L$ be a integer satisfying that $L \in [\log^{40} n, n]$.  Any randomized algorithm that is {\em oblivious to the start vertex} and given an $n$-vertex {\em directed} graph $G = (V,E)$ and a starting vertex $\ustart \in V$, samples from a distribution $\calD$ such that $\tvd{\calD}{\RW^G_L(\ustart)} \le 1 - \frac{1}{ \log^{10} n }$ requires  $\WT{\Omega}(n \cdot \sqrt{L})$ space.
\end{reminder}





The following inequality will be useful for the proof.
\begin{lemma}[{Fano's inequality (see, \eg,~\cite[Page~38]{CT06})}]\label{lemma:Fano}
	Let $\rv{Z}$ and $\rv{Z}'$ be two jointly distributed random variable over the same set $\mathcal{Z}$, it holds that
	\[
	\Pr[\rv{Z} \ne \rv{Z}'] \ge \frac{\H(\rv{Z} ~|~ \rv{Z}') - 1}{\log |\mathcal{Z}|} ~~\text{and}~~ \Pr[\rv{Z} = \rv{Z}'] \le \frac{\log |\calZ| - \H(\rv{Z} ~|~ \rv{Z}') + 1}{\log |\mathcal{Z}|}.
	\]
\end{lemma}

We will also need the following variant of the standard $\INDEX$ problem.

\begin{definition}[Multi-output generalization of $\INDEX$]\label{defi:INDEX-one-way}
	In the $\INDEX_{m,\ell}$ problem, Alice gets $\ell$ strings $X_1,\dotsc,X_\ell \in \{0,1\}^m$ and Bob gets an index $i \in [\ell]$. Alice sends a message to Bob and then Bob is required to output the string $X_i$.
\end{definition}

The lower bound for $\INDEX_{m,\ell}$ below will be crucial for our proof of~\autoref{theo:lowb-main-oblivious-algo}.

\begin{lemma}[One-way communication lower bound for $\INDEX_{m,\ell}$]\label{lemma:lowb-INDEX}
	Let $\calD_{\sf \INDEX}^{m,\ell}$ be the input distribution that Alice gets $\ell$ independent random strings each is uniformly distributed over $\{0,1\}^m$ and Bob gets a uniformly random index from $[\ell]$ that is independent of Alice's input. Solving $\INDEX_{m,\ell}$ over $\calD_{\sf \INDEX}$ with success probability at least $1/\log^{15} m$ requires Alice to send at least $m\ell/\log^{20} m$ bits to Bob.
\end{lemma}
\begin{proof}
	\newcommand{\rvX}{\rv{X}}
	\newcommand{\rvM}{\rv{M}}
	\newcommand{\rvI}{\rv{I}}
	Over the input distribution $\calD_{\sf \INDEX}^{m,\ell}$, Alice gets $\ell$ strings $\rvX_1,\dotsc,\rvX_{\ell}$, all distributed uniformly over $\{0,1\}^m$. Bob gets a uniformly random index $\rvI$ from $[\ell]$. By Yao's minimax theorem, to prove the theorem it suffices to bound the success probability of all deterministic one-way commutation protocols between Alice and Bob in which Alice sends at most $m\ell/\log^{20} m$ bits. In the following we fix such a protocol.
	
	Let $\rvM = \rvM(\rvX_1,\dotsc,\rvX_{\ell})$ be the message sent from Alice to Bob. We need the following claim.
	
	\begin{claim}\label{claim:mutual-info-bound}
		It holds that
		\[
		\E_{i \in [\ell]} \I(\rvX_i : \rvM) \le m/\log^{20} m.
		\]
	\end{claim}
	\begin{proof}
		For $i \in [\ell]$, let $\rvX_{<i} = (\rvX_1,\dotsc,\rvX_{i-1})$, by~\autoref{fact:infchain} and~\autoref{lemma:addcondition}, we have
		\[
		\sum_{i=1}^{\ell} \I(\rvX_i : \rvM) \le \sum_{i=1}^{\ell} \I(\rvX_i : \rvM ~|~ \rvX_{<i}) = \I(\rvX_1,\dotsc,\rvX_\ell : \rvM).
		\]
		Further noting that $\I(\rvX_1,\dotsc,\rvX_\ell : \rvM) \le |\rvM| \le m\ell/\log^{20} m$, the claim follows by taking an average.
	\end{proof}
	
	Now, let $f_i(M)$ be the (deterministic) output of Bob when receiving message $M$ from Alice and getting input $i$. Since Alice and Bob have independent inputs, the success probability of the protocol can be written as $\E_{i \in [\ell]} \Pr[f_i(\rvM) = \rvX_i]$.
	
	From the definition of mutual information, we have
	\[
	\H(\rvX_i ~|~ f_i(\rvM)) = \H(\rvX_i) - \I(\rvX_i : f_i(\rvM)) = m - \I(\rvX_i : f_i(\rvM)).
	\]
	
	Combing the above with~\autoref{lemma:Fano}, we have
	\begin{align*}
	\E_{i \in [\ell]} \Pr[f_i(\rvM) = \rvX_i]  &\le \frac{\E_{i \in [\ell]} \I(\rvX_i : f_i(\rvM)) + 1}{m} \\
											   &\le \frac{\E_{i \in [\ell]} \I(\rvX_i : \rvM) + 1}{m} \\
											   &\le 2/\log^{20} m \le 1/\log^{15} m,
	\end{align*}
	which completes the proof.
\end{proof}

Now we are ready to prove~\autoref{theo:lowb-main-oblivious-algo}.
\begin{proofof}{\autoref{theo:lowb-main-oblivious-algo}}
	Let $\tau = \sqrt{L}/\log^{20} n$. For a string $X \in \{0,1\}^{\tau^2}$ and $(i,j) \in [\tau]^2$, we use $X_{i,j}$ to denote the $((i-1)\tau + j)$-th bit in $X$. We will also need the following construction of gadget graphs.
	
	\begin{construction}{Gadget Construction $H_{\tau}(X)$}
		
		\begin{itemize} 

			\item \textbf{Setup:} Given a string $X \in \{0,1\}^{\tau^2}$.
			
			\item \textbf{Vertices:} We construct a layered graph $G$ with $2$ layers $V_1, V_2$ satisfying $\card*{V_1} = \tau$ and $\card*{V_2} = \tau + 1$. For convenience, we always use $V_{i,j}$ to denote the $j$-th vertex in the layer $V_i$.
			
			We will call $V_{2,\tau+1}$ as the starting vertex of $H_{\tau}(X)$.
			
			\item \textbf{Edges:} For every $(i,j) \in [\tau]^2$, we first add an edge from $V_{2,j}$ to $V_{1,i}$, and then also add an edge from $V_{1,i}$ to $V_{2,j}$ if $X_{i,j} = 1$. For every $i \in [\tau]$, we add an edge from $V_{1,i}$ to $V_{2,\tau + 1}$ and another edge from $V_{2,\tau + 1}$ back to $V_{1,i}$.
			
			\item \textbf{Edge ordering:} The edges are given in the lexicographically order. (Note that the edge ordering is not important for the lower bound, and we specify it only for concreteness.)
		\end{itemize}
	\end{construction}

\newcommand{\Xrec}{X_{\sf rec}}

	Now, let $m = \tau^2$ and $\ell = n/(2\tau + 1)$. Suppose there is a starting vertex oblivious streaming algorithm $\algo$ for simulating $L$-step random walks with space complexity $n \cdot \sqrt{L} / \log^{100} n$ and statistical distance at most $1 - 1/\log^{10} n$, we show it implies a one-way communication protocol solving $\INDEX_{m,\ell}$ over $\calD_{\sf \INDEX}^{m,\ell}$ that contradicts~\autoref{lemma:lowb-INDEX}.
	
	Let $\prealgo$ and $\sampalgo$ be the preprocessing subroutine and the sampling subroutine of $\algo$, respectively. The protocol is described as follows:
	
	\begin{construction}{Protocol $\Pi$ for $\INDEX_{m,\ell}$}
		\begin{enumerate} 
			\item Given input strings $X_1,\dotsc,X_\ell \in \{0,1\}^{\tau^2}$, Alice generates a graph $G = \bigsqcup_{i=1}^{\ell} H_{\tau}(X_i)$. That is, $G$ is an $n$-vertex graphs consists of $\ell$ clusters with the $i$-th cluster being $H_\tau(X_i)$. From now on, we will use $H_{\tau}(X_i)$ to denote the corresponding subgraph of $G$.
			
			\item Alice then simulates the preprocessing subroutine $\prealgo$ on the graph $G$, and sends its output to Bob.
			
			\item Bob gets an index $i \in [\ell]$ and sets the starting vertex in $H_\tau(X_i)$ to be $\ustart$. Bob then simulates $\sampalgo$ with Alice's message and $\ustart$ as the input to obtain a walk $W$.
			
			\item Given an $L$-step walk $\bar{W}$ starting from the starting vertex of $H_{\tau}(X_i)$, we define a string $\Xrec(\bar{W}) \in \{0,1\}^{\tau^2}$ as follows: let $V_1,V_2$ be the two layers in $H_\tau(X_i)$, $\Xrec(W)_{i,j} = 1$ if and only if $\bar{W}$ passes an edge from $V_{1,i}$ to $V_{2,j}$.
			
			\item Bob outputs $\Xrec(W)$.
		\end{enumerate}
	\end{construction}

	We first show that with high probability, an $L$-step random walk starting from the starting vertex in $H_{\tau}(X_i)$ determines the string $X_i$. Formally,  The following claim captures what we need.
	
	\begin{claim}\label{claim:recover}
		For every $X \in \{0,1\}^{\tau^2}$, letting $\ustart$ be the starting vertex of $H_\tau(X)$, it holds that
		\[
		\Pr_{W \sim \RW_{L}^{H_\tau(X)}(\ustart)}[\Xrec(W) = X] \ge 1 - \exp(-\log^{20} n).
		\]
	\end{claim}
	\begin{proof}
		We will first bound the probability that $\Xrec(W)_{i,j} \ne X_{i,j}$ for each $(i,j) \in [\tau]^2$ and then apply a union bound. Fix $(i,j) \in [\tau]^2$, if $X_{i,j} = 0$, since there is no edge from $V_{1,i}$ to $V_{2,j}$, clearly $\Xrec(W)_{i,j}$ is always $0$ as well. Hence we only need to consider the case that $X_{i,j} = 1$.
		
		In this case, one can observe that for every two steps, the random walk visits the edge between $V_{1,i}$ to $V_{2,j}$ with probability at least $1/(\tau + 1)^2$. And moreover, all these events are independent. 
		
		Hence, a random walk with $L = \tau^2 \cdot \log^{40} n$ steps visits the edge between $V_{1,i}$ to $V_{2,j}$ with probability
		\[
		1 - (1 - 1/(\tau + 1)^2)^{L/2} \ge 1 - \exp(-\Omega(1/\tau^2 \cdot L)) \ge 1 - \exp(-\Omega(\log^{40} n)).
		\]
		
		The claim then follows from a union bound.
	\end{proof}

	Finally, since our streaming algorithm $\algo$ has space complexity $n \cdot \sqrt{L} / \log^{100} n < m \ell / \log^{20} m$ and sampling error at most $1 - 1/\log^{10} n$. Protocol $\Pi$ also has communication complexity less than $m \ell / \log^{20} m$, and success probability at least $1/\log^{10} n - \exp(-\log^{20} n)$, which contradicts~\autoref{lemma:lowb-INDEX}.
\end{proofof}